\documentclass[runningheads]{llncs}
\usepackage{graphicx}
\usepackage{amsmath}
\let\printqed\qed
\renewcommand\qed{~\hspace*{\fill}$\printqed$}
\spnewtheorem{observation}{Observation}{\bfseries}{\rmfamily}
\newcommand{\nil}{\mathit{nil}}
\newcommand{\LSuf}{\mathit{lrs}}
\newcommand{\SSuf}{\mathit{sqs}}
\newcommand{\SPref}{\mathit{sqp}}

\newcommand{\STree}{\mathit{STree}}
\newcommand{\numocc}{\#\mathit{occ}}
\newcommand{\occ}{\mathit{occ}}
\newcommand{\MUS}{\mathsf{MUS}}
\newcommand{\StartToEnd}{\mathsf{S2E}}
\newcommand{\EndToStart}{\mathsf{E2S}}
\newcommand{\apoint}{\mathsf{pp}}
\newcommand{\spoint}{\mathsf{sp}}
\newcommand{\tpoint}{\mathsf{tp}}

\newcommand{\locus}{\mathit{locus}}
\newcommand{\rootnode}{\mathit{root}}
\newcommand{\depth}{\mathit{depth}}
\newcommand{\parent}{\mathit{par}}
\newcommand{\str}{\mathit{str}}
\newcommand{\start}{\mathit{start}}
\newcommand{\subtree}{\mathit{subtree}}
\newcommand{\hed}{\mathit{hed}}
\newcommand{\MAW}{\mathsf{MAW}}

\newcommand{\typeone}{\mathcal{M}_1}
\newcommand{\typetwo}{\mathcal{M}_2}
\newcommand{\typethree}{\mathcal{M}_3}
\begin{document}
\title{Minimal Unique Substrings and \\Minimal Absent Words in a Sliding Window}
\author{
  Takuya~Mieno\inst{1} \and
  Yuki~Kuhara\inst{1} \and
  Tooru~Akagi\inst{1} \and
  Yuta~Fujishige\inst{1,2} \and
  Yuto~Nakashima\inst{1} \and
  Shunsuke~Inenaga\inst{1} \and
  Hideo~Bannai\inst{1} \and
  Masayuki~Takeda\inst{1}
}
\authorrunning{T. Mieno et al.}
\institute{
  Department of Informatics, Kyushu University, Japan \and
  Japan Society for Promotion of Science, Japan\\
  \email{\texttt{\{takuya.mieno,yuki.kuhara,toru.akagi,yuta.fujishige,\\
                 yuto.nakashima,inenaga,bannai,takeda\}@inf.kyushu-u.ac.jp}
        }
}
\maketitle
\begin{abstract}
  A substring $u$ of a string $T$ is called a minimal unique substring~(MUS) of $T$
  if $u$ occurs exactly once in $T$ and any proper substring of $u$ occurs at least twice in $T$.
  A string $w$ is called a minimal absent word~(MAW) of $T$
  if $w$ does not occur in $T$ and any proper substring of $w$ occurs in $T$.
  In this paper,
  we study the problems of computing MUSs and MAWs in a sliding window over a given string $T$.
  We first show how the set of MUSs can change in a sliding window over $T$,
  and present an $O(n\log\sigma)$-time and $O(d)$-space algorithm to compute MUSs in a sliding window of width $d$ over $T$,
  where $\sigma$ is the maximum number of distinct characters in every window.
  We then give tight upper and lower bounds on the maximum number of changes in the set of MAWs
  in a sliding window over $T$.
  Our bounds improve on the previous results in [Crochemore et al., 2017].
\end{abstract}
 \section{Introduction}

Processing massive string data is a classical and important task
in theoretical computer science,
with a variety of applications such as data compression,
bioinformatics, and text data mining.
It is natural and common to assume
that such a massive string is given in an \emph{online} fashion,
one character at a time from left to right,
and that the memory usage is limited to some pre-determined space.
This is a so-called \emph{sliding window model},
where the task is to process
all substrings $T[i..i+d-1]$ of pre-fixed length $d$ in a string $T$ of length $n$
in an incremental fashion, for increasing $i = 1, \ldots, n-d+1$.
Usually the window size $d$ is set to be much smaller than the string length $n$,
and thus the challenge here is to design efficient algorithms that processes all such substrings using only $O(d)$ working space.
A typical application to the sliding window model is data compression;
examples are the famous Lempel-Ziv 77 (the original version)~\cite{LZ77}
and PPM~\cite{ClearyW84}.

In this paper, we study the following classes of strings in the sliding window model:
\emph{Minimal Unique Substrings} (\emph{MUSs})
and \emph{Minimal Absent Words} (\emph{MAWs}).
MUSs have been heavily utilized
for solving the \emph{Shortest Unique Substring} (\emph{SUS}) problem~\cite{Pei2013SUS,Tsuruta2014SUS,Hu2014SUS,Mieno2016SUSonRLE},
and MAWs have applications to data compression based on
\emph{anti-dictionaries}~\cite{Crochemore2000DCA,OtaFM14}.
However, despite the fact that
there is a common application field to MUSs and MAWs
such as bioinformatics~\cite{Graf2007MUS,Pei2013SUS,Chairungsee2012PhylogenyByMAW,Silva2015Ebola},
to our knowledge, these two objects were considered to be quite different and
were studied separately.
This paper is the first that brings a light to their similarities
by observing that a string $w$ is a MUS (resp. MAW) of a string $S$
if the number of occurrences of $w$ in $S$ is one (resp. zero),
and the number of occurrences of any proper substring of $w$
is at least two (resp. at least one).

We begin with combinatorial results on MUSs in a sliding window.
Namely, we show that the number of MUSs that are added or deleted
by one slide of the window is always $O(1)$ (Section~\ref{sec:comb_MUSs}).
We then present the first efficient algorithm that maintains
the set of MUSs for a sliding window of length $d$ over a string of length $n$
in a total of $O(n \log \sigma)$ time and $O(d)$ working space
(Section~\ref{sec:algo_MUSs}).
Our main algorithmic tool is the suffix tree for a sliding window that requires $O(d)$ space and can be maintained in $O(n \log \sigma)$ time~\cite{Larsson96,Senft2005suffix}.
Our algorithm for computing MUSs in a sliding window
is built on our combinatorial results,
and it keeps track of three different loci over the suffix tree,
all of which can be maintained in $O(\log \sigma)$ amortized time
per each sliding step.

MAWs in a sliding window have already been studied by
Crochemore et al.~\cite{Crochemore2017MAW}.
They studied the number of MAWs to be added / deleted
when the current window is shifted, and we improve some of these results
(Section~\ref{sec:comb_MAWs}):
For any string $T$ over an alphabet of size $\sigma$,
let $\MAW(T[i..j])$ be the set of all MAWs in the substring $T[i..j]$.
Crochemore et al.~\cite{Crochemore2017MAW} showed that
$|\MAW(T[i..i+d]) \setminus \MAW(T[i..i+d-1])| \leq (s_{i}-s_{\alpha})(\sigma -1) + \sigma + 1$
and
$|\MAW(T[i-1..i+d-1]) \setminus \MAW(T[i..i+d-1])| \leq (p_{i}-p_{\beta})(\sigma - 1) + \sigma + 1$,
where $s_{i}$, $s_{\alpha}$, $p_{i}$, and $p_{\beta}$ are
the lengths of the longest repeating suffix of $T[i..i+d-1]$,
of the longest suffix of $T[i..i+d-1]$ having an internal occurrence
immediately followed by $\alpha = T[i+d]$,
of the longest repeating prefix of $T[i..i+d-1]$,
and of the longest prefix of $T[i..i+d-1]$ having an internal occurrence
immediately preceded by $\beta = T[i-1]$.
Since both $s_{i}-s_{\alpha}$ and $p_{i}-p_{\beta}$ are in $\Theta(d)$
in the worst case, it leads to an $O(\sigma d)$ upper bound.
We improve this by showing that
both $|\MAW(T[i..i+d]) \setminus \MAW(T[i..i+d-1])|$
and $|\MAW(T[i-1..i+d-1]) \setminus \MAW(T[i..i+d-1])|$ are at most
$d + \sigma' + 1$, where $\sigma'$ is the number of
distinct characters in $T[i..i+d-1]$.
Since $\sigma' \leq d$, this leads to an improved $O(d)$ upper bound.
We also show that this is tight.
Crochemore et al.~\cite{Crochemore2017MAW} also showed that
$\sum_{i=1}^{n-d} |\MAW(T[i..i+d-1]) \bigtriangleup \MAW(T[i+1..i+d])| \in O(\sigma n)$.
We give an improved upper bound $O(\min\{\sigma, d\} n)$
and show that this is tight.

All the proofs omitted due to lack of space can be found in Appendix~\ref{sec:missing_proofs}.
 \section{Preliminaries}

\noindent \textbf{Strings.}
Let $\Sigma$ be an alphabet.
An element of $\Sigma$ is called a character.
An element of $\Sigma^\ast$ is called a string.
The length of a string $T$ is denoted by $|T|$.
The empty string $\varepsilon$ is the string of length 0.
If $T = xyz$, then $x$, $y$, and $z$ are called
a \emph{prefix}, \emph{substring}, and \emph{suffix} of $T$, respectively.
They are called a \emph{proper prefix}, \emph{proper substring},
and \emph{proper suffix} of $T$ if $x \neq T$, $y \neq T$, and $z \neq T$,
respectively.
If a string $b$ is a prefix of $T$ and is a suffix of $T$,
$b$ is called a \emph{border} of $T$.
For any $1 \le i \le |T|$, the $i$-th character of $T$ is denoted by $T[i]$.
For any $1 \le i \le j \le |T|$, $T[i..j]$ denote
the substring of $T$ starting at $i$ and ending at $j$.
For convenience, $T[i'..j'] = \varepsilon$ for $i' > j'$.
For any $1 \le i \le |T|$,
let $T[..i] = T[1..i]$ and $T[i..] = T[i..|T|]$.
For a string $w$, the set of beginning positions of occurrences of $w$ in $T$
is denoted by $\occ_T(w) = \{i \mid T[i..i+|w|-1] = w\}$.
Let $\numocc_T(w) = |\occ_T(w)|$.
For convenience, let $\numocc_T(\varepsilon) = |T|+1$.
In what follows, we consider an arbitrarily fixed string $T$ of length $n \ge 1$
over an alphabet $\Sigma$ of size $\sigma \ge 2$.

\noindent \textbf{Minimal Unique Substrings and Minimal Absent Words.}
Any string $w$ is said to be \emph{absent} from $T$ if $\numocc_T(w) = 0$,
and \emph{present} in $T$ if $\numocc_T(w) \geq 1$.
For any substring $w$ of $T$,
$w$ is called
\emph{unique} in $T$ if $\numocc_T(w) = 1$,
\emph{quasi-unique} in $T$ if $1 \leq \numocc_T(w) \leq 2$,
and \emph{repeating} in $T$ if $\numocc_T(w) \ge 2$.
A unique substring $w$ of $T$ is called \emph{a minimal unique substring}
of $T$ if any proper substring of $w$ is repeating in $T$.
Since a unique substring $w$ of $T$ has exactly one occurrence in $T$,
it can be identified with a unique interval $[s, t]$ such that
$1 \leq s \leq t \leq n$ and $w = T[s.. t]$.
We denote by
$\MUS(T) = \{[s, t] \mid T[s..t]~\text{is a MUS of}~T\}$
the set of intervals corresponding to the MUSs of $T$.
From the definition of MUSs, it is clear that
$[s, t] \in \MUS(T)$ if
(a) $T[s..t]$ is unique in $T$,
(b) $T[s+1..t]$ is repeating in $T$, and
(c) $T[s..t-1]$ is repeating in $T$.
See Fig.~\ref{fig:LS_SS_SP} in Appendix~\ref{sec:additional_fig} for an example of MUSs.

An absent string $w$ from $T$ is called a \emph{minimal absent word} of $T$ if
any proper substring of $w$ is present in $T$.
We denote by $\MAW(T)$ the set of all MAWs of $T$.
From the definition of MAWs,
it is clear that $w \in \MAW(T)$ if
(A) $w$ is absent from $T$,
(B) $w[2..]$ is present in $T$, and
(C) $w[..|w|-1]$ is present in $T$.

This paper deals with the problems of computing MUSs / MAWs in a sliding window
of fixed length $d$ over a given string $T$, formalized as follows:
\begin{description}
  \item[Input:] String $T$ of length $n$ and positive integer $d$~($< n$).
  \item[Output for the MUS problem:] $\MUS(T[i..i+d-1])$ for all $1 \leq i \leq n - d + 1$.
  \item[Output for the MAW problem:] $\MAW(T[i..i+d-1])$ for all $1 \leq i \leq n - d + 1$.
\end{description}

\noindent \textbf{Suffix trees.}
The \emph{suffix tree} of a string $T$, denoted $\STree_T$,
is a \emph{compacted trie} that represents all suffixes of $T$.
We consider a version of suffix trees a.k.a. \emph{Ukkonen trees}~\cite{Ukkonen1995}:
Namely, $\STree_T$ is a rooted tree such that
(1) each edge is labeled by a non-empty substring of $T$,
(2) each internal node has at least two children,
(3) the out-going edges of each node begin with mutually distinct characters,
(4) the suffixes of $T$ that are unique in $T$ are represented by
paths from the root to the leaves,
and the other suffixes of $T$ that are repeating in $T$ are represented by
paths from the root that end either on internal nodes or on edges.
To simplify the description of our algorithm,
we assume that there is an auxiliary node $\perp$
which is the parent of only the root node.
The out-going edge of $\perp$ is labeled with $\Sigma$;
This means that we can go down from $\perp$ by reading any character in $\Sigma$.
See Fig.~\ref{fig:STree} in Appendix~\ref{sec:additional_fig} for an example of $\STree_T$.

For each node $v$ in $\STree_T$,
$\parent(v)$ denotes the parent of $v$,
$\str(v)$ denotes the path string from the root to $v$,
$\depth(v)$ denotes the \emph{string depth} of $v$ (i.e. $\depth(v) = |\str(v)|$),
and $\subtree(v)$ denotes the subtree of $\STree_T$ rooted at $v$.
For each leaf $\ell$ in $\STree_T$,
$\start(\ell)$ denotes the starting position of $\str(\ell)$ in $T$.
For each non-empty substring $w$ of $T$,
$\hed(w) = v$ denotes the \emph{highest explicit descendant}
where $w$ is a prefix of $\str(v)$ and
$\depth(\parent(v)) < |w| \le \depth(v)$.
For each substring $w$ of $T$,
$\locus(w) = (u, h)$ represents the locus in $\STree_T$
where the path that spells out $w$ from the root terminates,
such that $u = \hed(w)$ and $h = \depth(u) - |w| \geq 0$.
We say that a substring $w$ of $T$ with $\locus(w) = (u, h)$
is represented by an \emph{explicit node} if $h = 0$,
and by an \emph{implicit node} if $h \geq 1$.
We remark that in the Ukkonen tree $\STree(T)$ of a string $T$,
some repeating suffixes may be represented by implicit nodes.
An implicit node which represents a suffix of $T$
is called an \emph{implicit suffix node}.
For any non-empty substring $w$ that is represented by an explicit node $v$,
the \emph{suffix link} of $v$ is a reversed edge
from $v$ to the explicit node that represents $w[2..]$.
The suffix link of the root that represents $\varepsilon$ points to $\perp$.
 \section{Combinatorial Results on MUSs in a Sliding Window} \label{sec:comb_MUSs}
Throughout this section, we consider
positions $i,j$~($1 \leq i \leq j \leq n$) such that $T[i..j]$ denotes
the sliding window for the $i$-th position over the input string $T$.
The following arguments hold for \emph{any} values of $i$ and $j$,
and hence, they will be useful for sliding windows of any length $d$.

Let $\LSuf_{i,j}$ be the longest repeating suffix of $T[i..j]$,
$\SSuf_{i,j}$ be the shortest quasi-unique suffix of $T[i..j]$,
and
$\SPref_{i,j}$ be the shortest quasi-unique prefix of $T[i..j]$.
Note that $\LSuf_{i,j}$ can be the empty string,
and that both $\SSuf_{i,j}$ and $\SPref_{i,j}$ are always non-empty strings.
See Fig.~\ref{fig:LS_SS_SP} in Appendix~\ref{sec:additional_fig} for examples.

The next lemmas
are useful for analyzing combinatorial properties on MUSs
and for designing an efficient algorithm for computing them in a sliding window.

\begin{lemma}\label{lem:LS_and_SS}
  The following three statements are equivalent:
  (1) $|\LSuf_{i,j}| \geq |\SSuf_{i,j}|$;
  (2) $\numocc_{T[i..j]}(\LSuf_{i,j}) = 2$;
  (3) $\numocc_{T[i..j]}(\SSuf_{i,j}) = 2$.
\end{lemma}

\begin{lemma}\label{lem:LS}
  $|\LSuf_{i,j+1}| \le |\LSuf_{i, j}| + 1$.
\end{lemma}

\subsection{Changes to MUSs when Appending a Character to the Right}
In this subsection, we consider an operation that slides the right-end
of the current window $T[i..j]$ with one character by appending the next character
$T[j+1]$ to $T[i..j]$. We use the following observation.
\begin{observation}\label{obs:suffix_occ}
  For each non-empty substring $s$ of $T[i..j]$,
  $\numocc_{T[i..j+1]}(s) \le \numocc_{T[i..j]}(s) + 1$.
  Also, $\numocc_{T[i..j+1]}(s) = \numocc_{T[i..j]}(s) + 1$
  if and only if $s$ is a suffix of $T[i..j+1]$.
\end{observation}
\subsubsection{MUSs to be Deleted when Appending a Character to the Right.}
Due to Observation~\ref{obs:suffix_occ},
we obtain Lemma~\ref{lem:deleted_MUS}
which describes MUSs to be deleted when a new character $T[j+1]$
is appended to the current window $T[i..j]$.
\begin{lemma}\label{lem:deleted_MUS}
  For any $[s, t]$ with $i \leq s < t \leq j$,
  $[s, t] \in \MUS(T[i..j])$ and $[s, t] \not\in \MUS(T[i..j+1])$ if and only if
  $T[s..t] = \SSuf_{i,j+1}$ and $\numocc_{T[i..j+1]}(\SSuf_{i,j+1}) = 2$.
\end{lemma}
\begin{proof}

  \noindent ($\Rightarrow$)
  Let $w = T[s.. t]$.
  Since $[s, t] \in \MUS(T[i..j])$ and $[s, t] \not\in \MUS(T[i..j+1])$,
  $\numocc_{T[i..j]}(w) = 1$ and $\numocc_{T[i..j+1]}(w) \ge 2$.
  It follows from Observation~\ref{obs:suffix_occ}
  that $\numocc_{T[i..j+1]}(w) = 2$ and $w$ is a suffix of $T[i..j+1]$.
If we assume that $w$ is a proper suffix of $\SSuf_{i,j+1}$,
  then $\numocc_{T[i..j+1]}(w) \ge 3$ by the definition of $\SSuf_{i,j+1}$,
  but this contradicts with $\numocc_{T[i..j+1]}(w) = 2$.
If we assume that $\SSuf_{i,j+1}$ is a proper suffix of $w$,
  then $\numocc_{T[i..j]}(\SSuf_{i,j+1}) \ge \numocc_{T[i..j]}(T[s+1..t]) \ge 2$.
  Also, $\numocc_{T[i..{j+1}]}(\SSuf_{i,j+1}) = \numocc_{T[i..j]}(\SSuf_{i,j+1}) + 1 \ge 3$
  by Observation~\ref{obs:suffix_occ},
  but this contradicts with the definition of $\SSuf_{i,j+1}$.
Therefore, we obtain $w = \SSuf_{i,j+1}$.
  Moreover, $\numocc_{T[i.. j+1]}(\SSuf_{i,j+1}) = 2$ since $w = \SSuf_{i, j+1}$ is a substring of $T[i..j]$.

  \noindent ($\Leftarrow$)
  Since $w = T[s..t]$ is a suffix of $T[i..j+1]$ and $\numocc_{T[i.. j+1]}(w) = 2$,
  $w$ is unique in $T[i..j]$.
  By the definition of $\SSuf_{j+1}$,
  a proper suffix $w[2..] = T[s+1..t]$ of $w = \SSuf_{i,j+1}$
  occurs at least three times in $T[i..j+1]$,
  i.e. $T[s+1..t]$ is repeating in $T[i..j]$~(see also Fig.~\ref{fig:deleted_MUS} in Appendix~\ref{sec:additional_fig} for illustration).
Also, a prefix $w[..|w|-1] = T[s..t-1]$ of $w = \SSuf_{i,j+1}$
  is clearly repeating in $T[i..j]$.
  Therefore, $w = T[s..t]$ is a MUS of $T[i..j]$ and is not a MUS of $T[i..j+1]$.
  \qed
\end{proof}

By Lemma~\ref{lem:deleted_MUS},
at most one MUS can be deleted
when appending $T[j+1]$ to the current window $T[i..j]$,
and such a deleted MUS must be $\SSuf_{i,j+1}$.

\subsubsection{MUSs to be Added when Appending a Character to the Right.}
First, we consider a MUS to be added when appending $T[j+1]$
to $T[i..j]$, which is a suffix of $T[i..j+1]$.
The next observation follows from the definition of $\LSuf_{i,j}$:
\begin{observation}\label{obs:suffix_MUS}
  If $[s, j] \in \MUS(T[i..j])$,
  then $s = j - |\LSuf_{i,j}|$.
  Namely, if there is a MUS of $T[i..j]$ that is a suffix of $T[i..j]$,
  then it must be the suffix of $T[i..j]$ that is exactly one character longer than $\LSuf_{i,j}$.
\end{observation}

\begin{lemma}\label{lem:LS_is_MUS}
  $[j+1 - k, j+1] \in \MUS(T[i..j+1])$ if and only if
  $T[j+1 - k.. j+1] = \alpha^{k+1}$ or $k \le |\LSuf_{i,j}|$,
  where $k = |\LSuf_{i,j+1}|$ and $\alpha = T[j+1]$.
\end{lemma}
\begin{proof}

  \noindent ($\Rightarrow$)
  Assume on the contrary that
  $T[j+1 - k.. j+1] \neq \alpha^{k+1}$ and
  $k > |\LSuf_{i, j}|$.
By the assumptions and Lemma~\ref{lem:LS}, $|\LSuf_{i, j}| = k - 1$,
  and thus, $T[j-|\LSuf_{i, j}|.. j] = T[j+1-k.. j]$.
  Since $T[j+1-k.. j+1]$ is a MUS of $T[i..j+1]$,
  $T[j+1-k.. j] = T[j-|\LSuf_{i, j}|.. j]$ occurs at least twice in $T[i..j+1]$.
  On the other hand, $T[j-|\LSuf_{i,j}|..j]$ is unique in $T[i..j]$ by the definition of $\LSuf_{i, j}$,
  hence $T[j-|\LSuf_{i, j}|..j]$ occurs in $T[i..j+1]$ as a suffix of $T[i..j+1]$.
  Consequently, we have $T[j-|\LSuf_{i, j}|.. j] = T[j+1-|\LSuf_{i, j}|.. j+1]$, i.e.
  $T[j-k.. j+1] = T[j+1-k.. j+1] = \alpha^{k+1}$ with $\alpha = T[j+1]$, a contradiction.

  \noindent ($\Leftarrow$)
  By definition,
  $T[j+2-k.. j+1] = \LSuf_{i,j+1}$ is repeating in $T[i..j+1]$
  and $T[j+1-k.. j+1]$ is unique in $T[i..j+1]$.
  Now it suffices to show $T[j+1-k.. j]$ is repeating in $T[i..j+1]$.
  If $T[j+1-k.. j+1] = \alpha^{k+1}$, then
  clearly $T[j+1-k.. j] = \alpha^k$ is repeating in $T[i..j+1]$.
  If $k \le |\LSuf_{i, j}|$, then
  $T[j+1-k.. j]$ is a suffix of $T[j+1 - |\LSuf_{i, j}|.. j]$~(see Fig.~\ref{fig:suffix_MUS} in Appendix~\ref{sec:additional_fig}).
  Thus $\numocc_{T[i..j+1]}(T[j+1-k.. j]) \ge \numocc_{T[i..j]}(T[j+1-k.. j])$
  $\ge \numocc_{T[i..j]}(T[j+1-|\LSuf_{i, j}|.. j]) \ge 2$.
  \qed
\end{proof}
Next, we consider MUSs to be added when appending $T[j+1]$ to $T[i..j]$,
which are \emph{not} suffixes of $T[i.. j+1]$.

\begin{lemma}\label{lem:MUS_contains_SS}
  For each~$[s, t] \in \MUS(T[i.. j+1])$ with $t \ne j+1$,
  if $[s, t] \not\in \MUS(T[i.. j])$ then
  $\numocc_{T[i..j+1]}(\SSuf_{i,j+1}) = 2$ and
  $\SSuf_{i, j+1}$ is a proper substring of $T[s..t]$.
\end{lemma}
\begin{proof}
  Since $t \ne j+1$, $T[s..t]$ is not a suffix of $T[i.. j+1]$.
  Moreover, since $[s, t] \in \MUS(T[i.. j+1])$, $T[s..t]$ is unique in $T[i..j]$.
  Since $T[s..t]$ is not a MUS of $T[i..j]$,
  there exists a MUS $u$ of $T[i..j]$
  which is a proper substring of $T[s..t]$.
  Assume on the contrary that $\numocc_{T[i..j+1]}(\SSuf_{i, j+1}) = 1$ or $u \neq \SSuf_{i, j+1}$.
  Then, it follows from Lemma~\ref{lem:deleted_MUS} that
  $u$ is a MUS of $T[i.. j+1]$.
  However, this contradicts with $[s, t] \in \MUS(T[i.. j+1])$.
  Therefore, $\numocc_{T[i..j+1]}(\SSuf_{i, j+1}) = 2$ and
  $u = \SSuf_{i, j+1}$ is a proper substring of $T[s..t]$.
  \qed
\end{proof}
Namely, a MUS which is not a suffix is added by appending one character
only if there is a MUS to be deleted by the same operation.
Moreover, such added MUSs must contain the deleted MUS.

\begin{lemma}\label{lem:new_MUSs}
  If $\numocc_{T[i..j+1]}(\SSuf_{i, j+1}) = 2$, then
  there are three integers $p_l, p_s, q$ such that $i \le p_l \le p_s \le q < j+1$
  and $T[p_s..q] = \SSuf_{i, j+1}$ and $T[p_l.. q] = \LSuf_{i, j+1}$.
  Also, the following propositions hold:
  \begin{itemize}
    \item[\normalfont{(a)}]
      If there is no MUS of $T[i..j]$ ending at $q+1$,
      then $[p_s, q+1] \in \MUS(T[i..j+1])$.
    \item[\normalfont{(b)}]
      If there is no MUS of $T[i..j]$ starting at $p_l-1$ and $p_l \ge i+1$,
      then $[p_l-1, q] \in \MUS(T[i..j+1])$.
  \end{itemize}
\end{lemma}

Now we have the main result of this subsection:
\begin{theorem}\label{thm:diff_MUS_add_right}
  For any $1 \leq i \leq j < n$,
  $|\MUS(T[i..j+1]) \bigtriangleup \MUS(T[i..j])| \le 4$ and
  $-1 \le |\MUS(T[i..j+1])| - |\MUS(T[i..j])| \le 2$.
Furthermore, these bounds are tight for any $\sigma, i, j$
with~$\sigma \ge 3$, $1 \le i \le j < n$, and $j-i+1 \ge 5$.
\end{theorem}

\subsection{Changes to MUSs when Deleting the Leftmost Character}
In this subsection, we consider an operation
that deletes the leftmost character $T[i-1]$ from $T[i-1..j]$.
Basically, we can use symmetric arguments to the previous subsection
where we considered appending a character to the right of the window.
We omit the details here in the case of deleting the leftmost character,
but all necessary observations and lemmas are available in Appendix~\ref{sec:theorem_two}.

The main result of this subsection is the following:

\begin{theorem}\label{thm:diff_MUS_delete_left}
  For any $1 < i \le j \le n$,
  $|\MUS(T[i-1..j]) \bigtriangleup \MUS(T[i..j])| \le 4$ and
  $-1 \le |\MUS(T[i-1..j])| - |\MUS(T[i..j])| \le 2$.
  Furthermore, these bounds are tight for any $\sigma, i, j$
  with~$\sigma \ge 3$, $1 < i \le j \le n$, and $j-i+1 \ge 5$.
\end{theorem}

The next corollary is immediate from Theorem~\ref{thm:diff_MUS_add_right} and Theorem~\ref{thm:diff_MUS_delete_left}.
\begin{corollary}\label{col:all_MUS_diff}
  Given a positive integer $d < n$.
  For every $i$ with~$1 \le i \le n-d$,
  $|\MUS(T[i..i+d-1]) \bigtriangleup \MUS(T[i+1..i+d])| \in O(1)$.
\end{corollary}
 \section{Algorithm for computing MUSs in a Sliding Window} \label{sec:algo_MUSs}
This section presents our algorithm for computing MUSs in a sliding window.
\subsection{Updating a Suffix Tree and Three Loci in a Suffix Tree}
First, we introduce some additional notions.
Since we use Ukkonen's algorithm~\cite{Ukkonen1995} for updating the suffix tree
when a new character $T[j+1]$ is appended to the right end of the window $T[i..j]$,
we maintain the locus for $\LSuf_{i, j}$ as in \cite{Ukkonen1995}.
Also, in order to compute the changes of MUSs,
we use $\SSuf_{i,j}$~(c.f. Lemma~\ref{lem:deleted_MUS} and Lemma~\ref{lem:new_MUSs}).
Thus, we also maintain the locus for $\SSuf_{i, j}$.

The locus for $\LSuf_{i,j}$~(resp.~$\SSuf_{i,j}$) in $\STree_{T[i..j]}$
is called the \emph{primary active point}~(resp.~the \emph{secondary active point})
and is denoted by $\apoint_{i,j}$~(resp.~$\spoint_{i,j}$).
Additionally, in order to maintain $\spoint_{i,j}$ efficiently,
we also maintain the locus for the longest suffix of $T[i..j]$ which occurs at least three times in $T[i..j]$.
We call this locus the \emph{tertiary active point} that is denoted by $\tpoint_{i,j}$.
\subsubsection{Appending One Character.}
When $T[i..j]$ is the empty string~(the base case, where $i = 1$ and $j = 0$),
we set all the three active points $(\rootnode, 0)$.
Then we increase $j$, and the suffix tree grows in an online manner until $j = d$
using Ukkonen's algorithm.
Then, for each $j > d$, we also increase $i$ each time $j$ increases,
so that the sliding window is shifted to the right,
by using sliding window algorithm for the suffix tree~\cite{Larsson96,Senft2005suffix}.

When $T[j+1]$ is appended to the right end of $T[i..j]$,
we first update the suffix tree to $\STree_{T[i,,j+1]}$ and compute $\apoint_{i, j+1}$.
Since $\apoint_{i, j+1}$ coincides with the \emph{active point},
$\apoint_{i, j+1}$ can be found in amortized $O(\log \sigma)$ time~\cite{Ukkonen1995,Larsson96,Senft2005suffix}.

After updating the suffix tree,
we can compute $\tpoint_{i, j+1}$ and $\spoint_{i, j+1}$ as follows:
\begin{enumerate}
  \item
    Traverse character $T[i+1]$ from $\tpoint_{i,j}$,
    and set $w \leftarrow \str(\tpoint_{i,j})T[i+1]$ which is the suffix of $T[i..j+1]$
    that is one character longer than $\tpoint_{i,j}$.
    Then, $w$ corresponds to a candidate for $\tpoint_{i, j+1}$.
  \item While $\numocc_{T[i..j+1]}(w) < 3$, set $w \leftarrow w[2..]$
    and search for the locus for $w$ by using suffix links in $\STree_{T[i..j+1]}$.
    This $w$ is a new candidate for $\tpoint_{i, j+1}$.
  \item After breaking the while-loop, obtain $\tpoint_{i,j+1} = \locus(w)$
    since $w$ is the longest suffix of $T[i..j+1]$ which occurs more than twice in $T[i..j+1]$.
  \item Also, $\spoint_{i,j+1}$ equals the locus
    which is the very previous candidate for $\tpoint_{i, j+1}$.
\end{enumerate}
As is described in the above algorithm, we can locate
$\tpoint_{i,j+1}$ using suffix link, in as similar manner to
the active point $\apoint_{i,j+1}$.
Thus, the cost for locating $\tpoint_{i,j+1}$ for each increasing $j$ is
amortized $O(\log \sigma)$ time,
again by a similar argument to the active point ($\apoint_{i,j+1}$).
What remains is, for each candidate $w$ for $\tpoint_{i,j+1}$,
how to quickly determine whether $\numocc_{T[i..j+1]}(w) < 3$ or not.
In what follows, we show that it can be checked in $O(1)$ time for each candidate.
\begin{observation}\label{obs:tp_occ}
  For each suffix $s$ of a string $T[i..j+1]$, let $\locus(s) = (u, h)$.
  \begin{description}
    \item[Case 1.]
      If $u$ is an internal node, $s$ occurs at least three times in $T[i..j+1]$.
    \item[Case 2.]
      If $u$ is a leaf and $h = 0$, $s$ occurs exactly once in $T[i..j+1]$.
    \item[Case 3.]\label{case_three}
      If $u$ is a leaf and $h \ne 0$,
      \begin{description}
        \item[Case 3.1.]
          if there is a suffix $s'$ of $T[i..j+1]$ with $\hed(s') = \hed(s)$ which is longer than $s$,
          $s$ occurs at least three times in $T[i..j+1]$~(see Fig.~\ref{fig:STree_occ} in Appendix~\ref{sec:additional_fig}).
        \item[Case 3.2.]
          otherwise, $s$ occurs exactly twice in $T[i..j+1]$.
      \end{description}
  \end{description}
\end{observation}

For any suffix $s$ of $T[i..j+1]$, if we are given $\locus(s) = (u, h)$,
then we can obviously determine in constant time whether
$s$ occurs at least three times in $T[i..j+1]$ or not,
except Case~\ref{case_three}.
The next lemma allows us to determine it in constant time
in Case~\ref{case_three}.
\begin{lemma} \label{lem:leaf_has_implicit_node}
  Suppose the locus $\apoint_{i, j+1}$ in $\STree_{T[i..j+1]}$ is already computed.
  Given a leaf $\ell$ of $\STree_{T[i..j+1]}$,
  it can be determined in $O(1)$ time whether
  there is an implicit suffix node on the edge~$(\parent(\ell), \ell)$
  and if so,
  the locus of the lowest implicit suffix node on $(\parent(\ell), \ell)$
  can be computed in $O(1)$ time.
\end{lemma}
\subsubsection{Deleting the Leftmost Character.}
When the leftmost character $T[i-1]$ is deleted from $T[i-1..j]$,
we first update the suffix tree and compute $\apoint_{i, j}$
by using the sliding window algorithm for the suffix tree~\cite{Larsson96,Senft2005suffix}.
Each pair of position pointers for the edge-labels of the suffix tree
can be maintained  in amortized $O(1)$ time so that these pointers always refer to positions
within the current sliding window, by
a simple \emph{batch update} technique~(see \cite{Senft2005suffix} for details).
After that, we compute $\tpoint_{i,j}$ and $\spoint_{i, j}$
in a similar way to the case of appending a new character shown previously.

It follows from the above arguments in this subsection
that we can update the suffix tree and the three active points in amortized~$O(\log \sigma)$ time, each time the window is shifted by one character.

\subsection{Computing $\SPref_{i-1,j}$}\label{subsec:sqs}
In order to compute the changes of MUSs when the leftmost character $T[i-1]$ is deleted from $T[i-1,j]$,
we use $\SPref_{i-1,j}$~(c.f. Lemma~\ref{lem:new_MUS} and Lemma~\ref{lem:deleted_MUSs})
before updating the suffix tree.
Thus, we present an efficient algorithm for computing $\SPref_{i-1,j}$.
First, we consider the following cases~(see Fig.~\ref{fig:SP_cases} in Appendix~\ref{sec:additional_fig}),
where $\ell$ is the leaf corresponding to $T[i-1..j]$:
\begin{description}
  \item[Case A.]
    $\hed(\LSuf_{i-1,j}) = \ell$.
  \item[Case B.]
    $\hed(\LSuf_{i-1,j}) \ne \ell$ and $\subtree(\parent(\ell))$ has more than two leaves.
  \item[Case C.]
    $\hed(\LSuf_{i-1,j}) \ne \ell$ and $\subtree(\parent(\ell))$ has exactly two leaves.
\end{description}

For Case A, the next lemma holds:
\begin{lemma}\label{lem:suffix_chain_come_back}
Given $\STree_{T[i-1..j]}$ and $\apoint_{i-1, j}$.
  Let $\ell$ be the leaf corresponding to $T[i-1..j]$.
  If $\apoint_{i-1,j}$ is on the edge $(\parent(\ell), \ell)$, the following propositions hold:
  \begin{itemize}
    \item[\normalfont{(a)}]
      $\occ_{T[i-1..j]}(\SPref_{i-1, j}) = \{i-1, j-|\LSuf_{i-1,j}|+1\}$.
    \item[\normalfont{(b)}]
      If there is exactly one implicit suffix node on $(\parent(\ell),\ell)$,
      $\SPref_{i-1,j} = T[i-1.. i-1+\depth(\parent(\ell))]$.
    \item[\normalfont{(c)}]
      If there are more than one implicit suffix node on $(\parent(\ell),\ell)$,
      then $|\LSuf_{i-1,j}| > \lfloor (j-i+2)/2 \rfloor$ and $\SPref_{i-1,j} = T[i-1.. j-2h+1]$,
      where $\apoint_{i-1,j} = (\ell, h)$.
  \end{itemize}
\end{lemma}
\begin{proof}
  Let $\apoint_{i-1,j} = (\ell, h)$ and $L =|\LSuf_{i-1,j}|$.
  \begin{itemize}
    \item[\normalfont{(a)}]
      Since $\apoint_{i-1,j}$ is on the edge $(\parent(\ell), \ell)$,
      $\SPref_{i-1, j}$ is a prefix of $\LSuf_{i-1,j}$, and
      $\numocc_{T[i-1..j]}(\LSuf_{i-1,j}) = \numocc_{T[i-1..j]}(\SPref_{i-1,j}) = 2$.
      Therefore, we obtain that $\occ_{T[i-1..j]}(\SPref_{i-1, j}) = \occ_{T[i-1..j]}(\LSuf_{i-1,j}) = \{i-1, j-L+1\}$.
    \item[\normalfont{(b)}]
      In this case, it is clear that $\SPref_{i-1,j} = T[i-1.. i-1+\depth(\parent(\ell))]$.
    \item[\normalfont{(c)}]
      Let $(\ell, h')$ be the locus of the implicit suffix node
      which is the lowest on the edge $(\parent(\ell),\ell)$
      except $\apoint_{i-1,j}$.
      Also, let $x$ be the string corresponding to the locus $(\ell, h')$.
      In this case, $x$ occurs exactly three times in $T[i-1..j]$.
      Also,
      $x$ is the longest border of $\LSuf_{i-1,j}$.
      Assume on the contrary that $L \le \lfloor (j-i+2)/2 \rfloor$.
      Then, two occurrences of $\LSuf_{i-1,j}$ in $T[i-1..j]$ are not overlapping, and thus
      $\numocc_{T[i-1..j]}(x) \ge 2 \times \numocc_{T[i-1..j]}(\LSuf_{i-1,j}) = 4$,
      it is a contradiction.
      Therefore, $L > \lfloor (j-i+2)/2 \rfloor$~(see Fig.~\ref{fig:SP} in Appendix~\ref{sec:additional_fig}).

      Next, we consider a relation between $h$ and $h'$.
      By the definition, $h = |T[i-1..j]| - L = j - i + 2 - L$.
      Since $L > \lfloor (j-i+2)/2 \rfloor$,
      $x$ matches the intersection of two occurrences of $\LSuf_{i-1,j}$,
      i.e. $x = T[j-L+1..i+L-2]$.
      Thus, $h' = |T[i-1..j]| - |x| = j-i+2-(2L-j+i-2) = 2(j-i+2-L) = 2h$.
      Therefore $\SPref_{i-1,j} = T[i-1..j-h'+1] = T[i-1.. j-2h+1]$.
      \qed
  \end{itemize}
\end{proof}

In Case B, it is clear that $\SPref_{i-1, j} = T[i-1.. i-1+\depth(p)]$
since $\str(p)$ occurs at least three times in $T[i-1..j]$~(see Fig.~\ref{fig:SP_cases} in Appendix~\ref{sec:additional_fig}).

For Case C, the next lemma holds:
\begin{lemma}\label{lem:two_children}
  Given $\STree_{T[i-1..j]}$ and $\apoint_{i-1, j}$.
  Let $\ell$ be the leaf corresponding to $T[i-1..j]$,
  $p = \parent(\ell)$, and $q = \parent(p)$.
If $\subtree(p)$ has exactly two leaves and
  there are no implicit suffix nodes on any edges in $\subtree(p)$,
  then it can be determined in $O(1)$ time
  whether there is an implicit suffix node on $(q, p)$.
If such an implicit node exists, then
  the locus of the lowest implicit suffix node on $(q, p)$ can be computed in $O(1)$ time.
\end{lemma}

We can design an algorithm for computing $\SPref_{i-1,j}$ by using the above lemmas,
as follows.
Let $\ell$ be the leaf corresponding to $T[i-1..j]$,
$p = \parent(\ell)$ and $q = \parent(p)$.
\begin{description}
  \item[In Case A.]
    $\SPref_{i-1, j}$ is computed by Lemma~\ref{lem:suffix_chain_come_back}.
  \item[In Case B.]
    $\SPref_{i-1, j} = T[i-1.. i-1+\depth(p)]$ and $\numocc_{T[i-1..j]}(\SPref_{i-1,j}) = 1$.
  \item[In Case C.]
    We divide this case into some subcases by the existence of
    an implicit suffix node on edges $(p, \ell')$ and $(q, p)$
    where $\ell'$ is the sibling of $\ell$.
    We first determine the existence of an implicit suffix node
    on $(p, \ell')$~(by Lemma~\ref{lem:leaf_has_implicit_node}).
  \begin{itemize}
    \item If there is an implicit suffix node on $(p, \ell')$, then
      $\SPref_{i-1,j} = T[i-1..i-1 + \depth(p)]$ and $\numocc_{T[i-1..j]}(\SPref_{i-1,j}) = 1$.
    \item If there is no implicit suffix node on both $(p, \ell)$ and $(p, \ell')$,
      we can determine in constant time
      the existence of an implicit suffix node on $(q, p)$~(by Lemma~\ref{lem:two_children}).
If there is an implicit suffix node on $(q, p)$,
      $\SPref_{i-1,j} = T[i-1..\depth(p)-h+1]$ and $\occ_{T[i-1..j]}(\SPref_{i-1,j}) = \{i-1, \start(\ell')\}$.
      Otherwise,
      $\SPref_{i-1,j} = T[i-1.. \depth(q) + 1]$ and $\occ_{T[i.-1.j]}(\SPref_{i-1,j}) = \{i-1, \start(\ell')\}$.
  \end{itemize}
\end{description}

It follows from the above arguments in this subsection
that $\SPref_{i-1, j}$ can be computed in $O(1)$ time
by using the suffix tree and the (primary) active point.

\subsection{Detecting MUSs to be Added/Deleted}
By using the afore-mentioned lemmas in this section,
we can design an efficient algorithm for detecting MUSs to be added / deleted.
The details of our algorithm can be found in Appendix~\ref{sec:detect_MUS}.

The main result of this section is the following:
\begin{theorem}\label{thm:sliding_window_MUS}
  We can maintain the set of MUSs in a sliding window
  of length $d$ on a string $T$ of length $n$ over an alphabet of size $\sigma$,
  in a total of $O(n\log \sigma)$ time and $O(d)$ working space.
\end{theorem}
\begin{corollary}\label{col:online_MUS}
  There exists an \emph{online} algorithm to
  compute all MUSs in a string $T$
  of length $n$ over an alphabet of size $\sigma$
  in a total of $O(n\log \sigma)$ time with $O(n)$ working space.
\end{corollary}
 \section{Combinatorial Results on MAWs in a Sliding Window} \label{sec:comb_MAWs}

\subsection{Changes to MAWs when Appending Character to the Right}
We consider the number of changes of MAWs when appending
$T[j+1]$ to $T[i..j]$.

For the number of deleted MAWs, the next lemma is known:
\begin{lemma}[\cite{Crochemore2017MAW}] \label{lem:extention_dec}
  For any $1 \le i \le j < n$,
  $|\MAW(T[i..j])\setminus \MAW(T[i..j+1])| = 1$.
\end{lemma}

Next, we consider the number of added MAWs.
We classify each MAW $w$ in $\MAW(T[i..j+1]) \setminus \MAW(T[i..j])$
to the following three types\footnote{At least one of $w[2..]$ and $w[..|w|-1]$ is absent from $T[i..j]$,
because $w \not\in \MAW(T[i..j])$.}~(see Fig.~\ref{fig:MAW_3types} in Appendix~\ref{sec:additional_fig}).
Let $\sigma'$ be the number of distinct characters occurring in $T[i..j]$.
\begin{description}
  \item[Type 1.]
    $w[2..]$ and $w[..|w|-1]$ are both absent from $T[i..j]$.
  \item[Type 2.]
    $w[2..]$ is present in $T[i..j]$ and $w[..|w|-1]$ is absent from $T[i..j]$.
  \item[Type 3.]
    $w[2..]$ is absent from $T[i..j]$ and $w[..|w|-1]$ is present in $T[i..j]$.
\end{description}
We denote by $\typeone$, $\typetwo$, and $\typethree$ the set of MAWs of
Type 1, Type 2 and Type 3, respectively.
The next lemma holds:
\begin{lemma}\label{lem:extention_inc}
  For any $1 \le i \le j < n$,
  $|\MAW(T[i..j+1]) \setminus \MAW(T[i..j])| \le \sigma' + d$, where
  $d = j - i + 1$.
\end{lemma}
\begin{proof}
  In \cite{Crochemore2017MAW}, it is shown that $|\typeone| \le 1$.
  It is also shown in \cite{Crochemore2017MAW} that
  the last characters of all MAWs in $\typetwo$ are all different.
  Furthermore, by the definition of $\typetwo$,
  the last character of each MAW in $\typetwo$ occurs in $T[i..j]$.
  Thus, $|\typetwo| \le \sigma'$.
  In the rest of the proof, we show that the number of MAWs of Type 3 is at most $d-1$.
  We show that there is an injection $f: \typethree \rightarrow [i, j-1]$
  that maps each MAW $w \in \typethree$ to
  the ending position of the leftmost occurrence of $w[..|w|-1]$ in $T[i..j]$.
By the definition of $\typethree$,
  $w$ is absent from $T[i..j+1]$ and $w[|w|] = T[j+1]$
  for each $w \in \typethree$,
  and thus, no occurrence of $w[..|w|-1]$ in $T[i..j]$
  ends at position $j$.
  Hence, the range of $f$ does not contain the position $j$, i.e. it is $[i..j-1]$.
Next, for the sake of contradiction,
  we assume that $f$ is not an injection,
  i.e. there are two distinct MAWs $w_1, w_2 \in \typethree$ such that $f(w_1) = f(w_2)$.
  W.l.o.g., assume $|w_1| \ge |w_2|$.
  Since $w_1[|w_1|] = w_2[|w_2|] = T[j+1]$ and $f(w_1) = f(w_2)$, $w_2$ is a suffix of $w_1$.
  If $|w_1| = |w_2|$, then $w_1 = w_2$ and it contradicts with $w_1 \neq w_2$.
  If $|w_1| > |w_2|$, then $w_2$ is a proper suffix of $w_1$, and it contradicts
  with the fact that $w_2$ is absent from $T[i..j+1]$~(see Fig.~\ref{fig:MAW_Type3} in Appendix~\ref{sec:additional_fig}).
  Therefore, $f$ is an injection and $|\typethree| \le j-1-i+1 = d-1$.
  \qed
\end{proof}

The next lemma follows from Lemma~\ref{lem:extention_dec} and Lemma~\ref{lem:extention_inc}.

\begin{lemma}\label{lem:extention_total}
  For any $1 \le i \le j < n$,
  $|\MAW(T[i..j+1]) \bigtriangleup \MAW(T[i..j])| \le \sigma' + d + 1$, where
  $d = j-i+1$.
  The upper bound is tight
  when $\sigma \ge 3$ and $\sigma' + 1 \le \sigma$.
\end{lemma}

\subsection{Changes to MAWs when Deleting the Leftmost Character}
Next, we analyze the number of changes of MAWs when deleting the leftmost character from a string.
By a symmetric argument to Lemma~\ref{lem:extention_total}, we obtain the next lemma:
\begin{lemma}\label{lem:reduction_total}
  For any $1 < i \le j \le n$,
  $|\MAW(T[i.. j]) \bigtriangleup \MAW(T[i-1.. j])| \le \sigma' + d + 1$ where
  $d = j-i+1$ and
  $\sigma'$ is the number of distinct characters occurs in $T[i.. j]$.
  Also, the upper bound is tight
  when $\sigma \ge 3$ and $\sigma' + 1 \le \sigma$.
\end{lemma}
Finally, by combining Lemma~\ref{lem:extention_total} and Lemma~\ref{lem:reduction_total},
we obtain the next corollary:
\begin{corollary}\label{col:increase_MAW_at_one_slide}
  Let $d$ be the window length.
  For a string $T$ of length $n > d$ and each integer $i$ with $1 \le i \le n-d$,
  $|\MAW(T[i..i+d-1]) \bigtriangleup \MAW(T[i+1..i+d])| \in O(d)$.
  Also, there exists a string $T'$ which satisfies
  $|\MAW(T'[j..j+d-1]) \bigtriangleup \MAW(T'[j+1..j+d])| \in \Omega(d)$
  for some $j$ with $1 \le j \le |T'|-d$.
\end{corollary}

\subsection{Total Changes of MAWs when Sliding the Window on a String}
In this subsection, we consider the total number of changes of MAWs
when sliding the window of length $d$ from the beginning of $T$ to the end of $T$.
We denote the total number of changes of MAWs by
$\mathcal{S}(T, d) = \sum_{i=1}^{n-d}|\MAW(T[i..i+d-1]) \bigtriangleup \MAW(T[i+1..i+d])|$.
The following lemma is known:
\begin{lemma}[\cite{Crochemore2017MAW}] \label{lem:sigma_slide_MAW_Crochemore}
  For a string $T$ of length $n > d$ over an alphabet $\Sigma$ of size $\sigma$,
  $\mathcal{S}(T, d) \in O(\sigma n)$.
\end{lemma}

The aim of this subsection is to give a more rigorous bound for $\mathcal{S}(T, d)$.
We first show that the above bound is tight under some conditions.
\begin{lemma}\label{lem:sigma_slide_MAW}
  The upper bound of Lemma~\ref{lem:sigma_slide_MAW_Crochemore}
  is tight when $\sigma \le d$ and $n-d \in \Omega(n)$
  (see Appendix~\ref{sec:missing_proofs} for a proof).
\end{lemma}

Next, we consider the case where $\sigma \ge d+1$.
\begin{lemma}\label{lem:d_slide_MAW}
  For a string $T$ of length $n > d$ over an alphabet $\Sigma$ of size $\sigma$,
  $\mathcal{S}(T, d) \in O(d(n-d))$,
  and this upper bound is tight when $\sigma \ge d+1$
  (see Appendix~\ref{sec:missing_proofs} for a proof).
\end{lemma}

The main result of this section follows from the above lemmas:
\begin{theorem}\label{thm:total_slide_MAW}
  For a string $T$ of length $n > d$ over an alphabet $\Sigma$ of size $\sigma$,
  $\mathcal{S}(T, d) \in O(\min\{d, \sigma\} n)$.
  This upper bound is tight when $n-d \in \Omega(n)$.
\end{theorem}
We remark that $n-d \in \Omega(n)$ covers most interesting cases
for the window length $d$, since the value of $d$ can range
from $O(1)$ to $cn$ for any $0 < c < 1$.
 \bibliographystyle{splncs04}
\bibliography{ref}
\clearpage
\appendix
\section{Proofs for Theorem~\ref{thm:diff_MUS_delete_left}} \label{sec:theorem_two}
In this appendix, we provide omitted proofs for Theorem~~\ref{thm:diff_MUS_delete_left}.
\begin{observation}\label{obs:prefix_occ}
  For each non-empty substring $s$ of $T[i-1..j]$,
  $\numocc_{T[i-1..j]}(s) \le \numocc_{T[i..j]}(s) + 1$.
  Also, $\numocc_{T[i-1..i]}(s) = \numocc_{T[i..j]}(s) + 1$
  if and only if $s$ is a prefix of $T[i-1..j]$.
\end{observation}
\subsubsection{MUSs to be Added when Deleting the Leftmost Character.}
\begin{lemma}\label{lem:new_MUS}
  For any $i \leq s \leq t \leq j$,
  $[s, t] \not\in \MUS(T[i-1..j])$ and $[s, t] \in \MUS(T[i..j])$ if and only if
  $T[s..t] = \SPref_{i-1,j}$ and $\numocc_{T[i-1..j]}(\SPref_{i-1,j}) = 2$.
\end{lemma}
\begin{proof}
  Symmetric to the proof of Lemma~\ref{lem:deleted_MUS}.
  \qed
\end{proof}
\subsubsection{MUSs to be Deleted when Deleting the Leftmost Character.}
Next, we consider MUSs to be deleted by removing $T[i-1]$ from $T[i-1..j]$.
If there is a MUS $w$ of $T[i-1..j]$ which is a prefix of $T[i-1..j]$,
clearly, $w$ is not a MUS of $T[i..j]$.
Then, we consider MUSs to be deleted each of which are \emph{not} a prefix of $T[i-1..j]$.

\begin{lemma}\label{lem:deleted_MUSs_contain_SP}
  For each $[s, t] \in \MUS(T[i-1, j])$ with $s \ne i-1$,
  if $[s, t] \not\in \MUS(T[i..j])$
  then $\numocc_{T[i-1..j]}(\SPref_{i-1,j}) = 2$ and
  $\SPref_{i-1,j}$ is a proper substring of $T[s..t]$.
\end{lemma}
\begin{proof}
  Symmetric to the proof of Lemma~\ref{lem:MUS_contains_SS}.
  \qed
\end{proof}
Namely, when deleting the leftmost character,
a MUS which is not a prefix is deleted
only if an added MUS exists.
Moreover, such deleted MUSs must contains the added MUS.
\begin{lemma}\label{lem:deleted_MUSs}
  If $\numocc_{T[i-1..j]}(\SPref_{i-1,j}) = 2$, then
  following propositions hold:
  \begin{itemize}
    \item[\normalfont{(a)}]
      If there is a MUS $w$ starting at $s$ in $T[i-1..j]$,
      $w$ is not a MUS of $T[i..j]$,
    \item[\normalfont{(b)}]
      If there is a MUS $w'$ ending at $t$ in $T[i-1..j]$,
      $w'$ is not a MUS of $T[i..j]$,
  \end{itemize}
  where $T[s..t] = \SPref_{i-1,j}$ and $s \ne i-1$.
\end{lemma}
\begin{proof}
  Symmetric to the proof of Lemma~\ref{lem:new_MUSs}.
  \qed
\end{proof}
\begin{proof}[of Theorem~\ref{thm:diff_MUS_delete_left}]
  Symmetric to the proof of Theorem~\ref{thm:diff_MUS_add_right}.
  \qed
\end{proof}
\clearpage
\section{Detecting MUSs to be Added/Deleted}\label{sec:detect_MUS}
In this appendix, we present our algorithm for detecting MUSs to be added / deleted.

\subsubsection{Data Structure for Maintaining MUSs.}
First, we introduce a data structure for managing the set of MUSs in a sliding window.
Our data structure for MUSs consists of two arrays $\StartToEnd$ and $\EndToStart$ of length $d$ each.
Note that by the definition of MUSs, any MUSs cannot be nested each other.
Thus, for any text position $i$, if a MUS starting~(resp. ending) at $i$ exists,
then its ending~(resp. starting) position is a unique.
From this fact, we can define $\StartToEnd$ and $\EndToStart$ as follows:

Let $[p, p+d-1]$ be the current window.
For every index $i$ with~$p \le i \le p+d-1$,
\begin{equation*}
  \StartToEnd[(i-1)~\text{mod}~d + 1] =
  \begin{cases}
    e & \text{if $[i,e]\in\MUS(T[p..p+d-1])$ exists,}\\
    \nil & \text{otherwise.}
  \end{cases}
\end{equation*}
\begin{equation*}
  \EndToStart[(i-1)~\text{mod}~d + 1] =
  \begin{cases}
    s & \text{if $[s,i]\in\MUS(T[p..p+d-1])$ exists,}\\
    \nil & \text{otherwise}.
  \end{cases}
\end{equation*}
Since MUSs cannot be nested each other,
these arrays are uniquely defined~(see Fig.~\ref{fig:MUS_array}).
By using these two arrays, all the following operations for MUSs
can be executed in $O(1)$ time;
add/remove a MUS into/from the set of MUSs, and
compute the ending/starting position of the MUS that starts/ends at a specified position.
\begin{figure}[ht]
    \centerline{\includegraphics[width=0.8\linewidth]{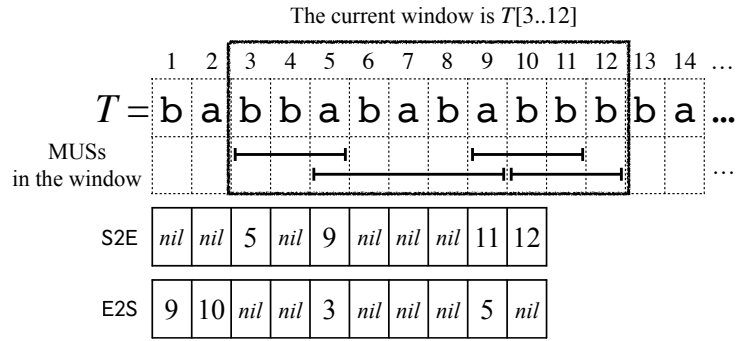}}
    \caption{
      A long string $T = \mathtt{babbabababbbba}\cdots$ and two arrays $\StartToEnd$ and $\EndToStart$.
      The current window is $T[3..12]$ of length $d = 10$, and the MUSs in the window are
      $T[3..5], T[5..9], T[9..11]$, and $T[10..12]$.
    }\label{fig:MUS_array}
\end{figure}
\subsubsection{Algorithm when Appending a Character to the Right.}
Assume that $\StartToEnd$, $\EndToStart$ and the suffix tree of $T[i..j]$ are computed
before reading $\gamma = T[j+1]$.
Also, assume that the longest single character run $\beta^e$ as a suffix of $T[i..j]$ is known,
where $\beta = T[j]$ and $e \ge 1$.
\begin{itemize}
  \item First, compute the length of $\LSuf_{i, j}$.
  \item Second, read $\gamma$, and update the suffix tree and the active points.
    Then, compute the lengths of $\LSuf_{i, j+1}$ and $\SSuf_{i, j+1}$.
    Also, update information about the run of the last character of $T[i..j+1]$.
    Specifically, if $\gamma = \beta$ then $\beta^e = \gamma^{e+1}$,
    and otherwise $\beta^e = \gamma^1$.
    If $|\LSuf_{i, j+1}| \le |\LSuf_{i, j}|$ or $T[j+1-|\LSuf_{i, j+1}|..j+1] = \gamma^{|\LSuf_{i, j+1}|+1}$,
    add $[j+1-|\LSuf_{i, j+1}|, j+1]$ into the set of MUSs~(by Lemma~\ref{lem:LS_is_MUS}).
  \item If $|\LSuf_{i, j+1}| < |\SSuf_{i, j+1}|$, then terminate this step~(by Lemma~\ref{lem:MUS_contains_SS}).
  \item Otherwise, compute $p_s$ and $q$ of Lemma~\ref{lem:new_MUSs} by using $\STree(T[i..j+1])$ and $\spoint_{i, j+1}$.
    Then, remove $[p_s, q]$ from the set of MUSs~(by Lemma~\ref{lem:deleted_MUS}).
  \item Next, if $\EndToStart[(t'-1)~\text{mod}~d+1] = \nil$,
    then add $[p_s, t']$ into the set of MUSs, where $t' = q + 1$.
    Also, if $s' \ge i$ and $\StartToEnd[(s'-1)~\text{mod}~d-1] = \nil$,
    then add $[s', q]$ into the set of MUSs,
    where $s' = q - |\LSuf_{i, j+1}|$~(by Lemma~\ref{lem:new_MUSs}).
  \item Terminate this step.
\end{itemize}

\subsubsection{Algorithm when Deleting the Leftmost Character.}
Assume that $\StartToEnd$, $\EndToStart$ and the suffix tree of $T[i-1..j]$ are computed
before deleting $\alpha = T[i-1]$.
\begin{itemize}
  \item First, compute $\numocc_{T[i-1..j]}(\SPref_{i-1,j})$.
    If $\numocc_{T[i-1..j]}(\SPref_{i-1,j}) = 2$,
    compute two integers $s$ and $t$ with $T[s..t] = \SPref_{i-1,j}$ and $s \ne i-1$.
  \item Second, delete $T[i-1]$ and update the suffix tree and the active points.
    If $\StartToEnd[(i-1-1)~\text{mod}~d+1] \ne \nil$, remove the MUS starting at $i-1$ from the set of MUSs.
  \item If $\numocc_{T[i-1..j]}(\SPref_{i-1,j}) = 1$, terminate this step~(by Lemma~\ref{lem:deleted_MUSs_contain_SP}).
  \item Otherwise, if $\StartToEnd[(s-1)~\text{mod}~d+1] \ne \nil$,
    then remove the MUS starting at $s$ from the set of MUSs.
    Also, if $\EndToStart[(t-1)~\text{mod}~d+1] \ne \nil$,
    then remove the MUS ending at $t$ from the set of MUSs~(by Lemma~\ref{lem:deleted_MUSs}).
  \item Finally, add $[s, t]$ into the set of MUSs~(by Lemma~\ref{lem:new_MUS}), and
    terminate this step.
\end{itemize}
\clearpage
\section{Omitted Proofs} \label{sec:missing_proofs}
In this appendix, we provide proofs that are omitted due to lack of space.
\begin{proof}[of Lemma~\ref{lem:LS_and_SS}]

  \noindent \textbf{(1) $\Rightarrow$ (2) and (3):}
  Since $|\LSuf_{i,j}| \geq |\SSuf_{i,j}|$, $\SSuf_{i,j}$ is a suffix of $\LSuf_{i,j}$
  and thus $\numocc_{T[i..j]}(\SSuf_{i,j}) \ge \numocc_{T[i..j]}(\LSuf_{i,j})$.
  By the definitions of $\SSuf_{i,j}$ and $\LSuf_{i,j}$,
  $\numocc_{T[i..j]}(\SSuf_{i,j}) \le 2$ and $\numocc_{T[i..j]}(\LSuf_{i,j}) \ge 2$.
  Thus $\numocc_{T[i..j]}(\LSuf_{i,j}) = \numocc_{T[i..j]}(\SSuf_{i,j}) = 2$.

  \noindent \textbf{(2) $\Rightarrow$ (1):}
  Since $\numocc_{T[i..j]}(\LSuf_{i,j}) = 2$,
  the shortest suffix $\SSuf_{i,j}$ of $T[i..j]$ that occurs at most twice in $T[i..j]$
  cannot be longer than $\LSuf_{i,j}$, i.e. $|\LSuf_{i,j}| \geq |\SSuf_{i,j}|$.

  \noindent \textbf{(3) $\Rightarrow$ (1):}
  Since $\numocc_{T[i..j]}(\SSuf_{i,j}) = 2$,
  the longest suffix $\LSuf_{i,j}$ of $T[i..j]$ that occurs at least twice in $T[i..j]$
  is at least as long as $\SSuf_{i,j}$, i.e. $|\LSuf_{i,j}| \geq |\SSuf_{i,j}|$.
  \qed
\end{proof}
\begin{proof}[of Lemma~\ref{lem:LS}]
  Assume on the contrary that $|\LSuf_{i,j+1}| > |\LSuf_{i,j}| + 1$.
  By the definition of $\LSuf_{i,j+1}$,
  $\LSuf_{i,j+1} = T[j+2-|\LSuf_{i,j+1}|.. j+1]$ occurs at least twice in $T[i..j+1]$.
  Hence, $T[j+2-|\LSuf_{i,j+1}|.. j]$ which is a proper prefix of $\LSuf_{i,j+1}$
  also occurs at least twice in $T[i..j]$.
  In addition,
  $\LSuf_{i,j} = T[j+2-|\LSuf_{i,j}|..j]$ is a proper suffix of $T[j+2-|\LSuf_{i,j+1}|.. j]$
  since $|\LSuf_{i,j+1}| > |\LSuf_{i,j}| + 1$.
  However, this contradicts the definition of $\LSuf_{i,j}$.
  Therefore, $|\LSuf_{i,j+1}| \le |\LSuf_{i,j}| + 1$.
  \qed
\end{proof}
\begin{figure}[ht]
  \centerline{\includegraphics[width=0.8\linewidth]{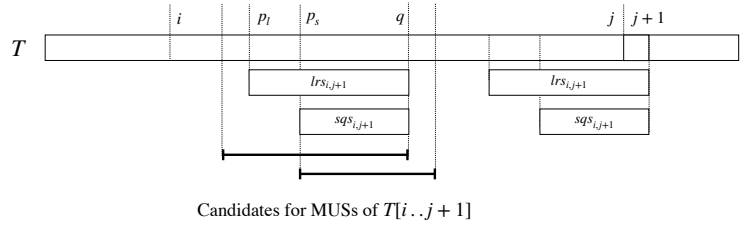}}
    \caption{
      Illustration of the situation
      when $\SSuf_{i, j+1}$ is repeating in $T[i..j+1]$.
      In this situation, $[p_l-1, q]$ and $[p_s, q+1]$ are the only candidates for
      MUSs in $\MUS(T[i..j+1]) \setminus \MUS(T[i..j])$
      each of which is not a suffix of $T[i..j+1]$.
    }\label{fig:MUS_candidates}
\end{figure}
\begin{proof}[of Lemma~\ref{lem:new_MUSs}]
  By Lemma~\ref{lem:LS_and_SS},
  $\numocc_{T[i..j+1]}(\LSuf_{i, j+1}) = 2$ and $\SSuf_{i, j+1}$ is a suffix of $\LSuf_{i, j+1}$
  since $\numocc_{T[i..j+1]}(\SSuf_{i,j+1}) = 2$.
  Hence,
  the ending position of the occurrence of $\SSuf_{i, j+1}$ in $T[i..j]$ and
  that of $\LSuf_{i, j+1}$ in $T[i..j]$
  are the same~(see Fig.~\ref{fig:MUS_candidates}).
  Next, we consider MUSs to be added.
  \begin{itemize}
    \item[(a)]
      For the sake of contradiction, we assume that $T[p_s..q+1]$ is repeating in $T[i..j+1]$, then
      $\numocc_{T[i..j+1]}(T[p_s..q]) \ge 3$, and it contradicts
      the definition of $\SSuf_{i,j+1}~(= T[p_s..q])$.
      Hence, $T[p_s..q+1]$ is unique in $T[i..j+1]$.
Also, $T[p_s..q] = \SSuf_{i, j+1}$ is repeating in $T[i..j+1]$ clearly.
      In addition, $T[p_s+1..q+1]$ is repeating in $T[i..j]$
      since $[p_s, q] \in \MUS(T[i..j])$~(by Lemma~\ref{lem:deleted_MUS})
      and there is no MUS of $T[i..j]$ ending at $q+1$.
      Thus, $T[p_s+1..q+1]$ is also repeating in $T[i..j+1]$.
      Therefore, $T[p_s..q+1]$ is a MUS of $T[i..j+1]$.
    \item[(b)]
      For the sake of contradiction, we assume that $T[p_l-1..q]$ is repeating in $T[i..j+1]$.
      Then $T[p_l-1] = T[j+1-|\LSuf_{i,j+1}|]$ or $\numocc_{T[i..j+1]}(T[p_l..q]) \ge 3$.
      If $T[p_l-1] = T[j+1-|\LSuf_{i, j+1}|]$, it contradicts the definition of $\LSuf_{i, j+1}$.
      If $\numocc_{T[i..j+1]}(T[p_l..q]) \ge 3$, it contradicts $\numocc_{T[i..j+1]}(\LSuf_{i, j+1}) = 2$.
      Thus, $T[p_l-1..q]$ is unique in $T[i..j+1]$.
Also, $T[p_l..q] = \LSuf_{i, j+1}$ is repeating in $T[i..j+1]$ clearly.
      In addition, $T[p_l-1..q-1]$ is repeating in $T[i..j]$,
      since $[p_s, q] \in \MUS(T[i..j])$~(by Lemma~\ref{lem:deleted_MUS}),
      $T[p_l..q-1]$ is repeating in $T[i..j]$, and
      there is no MUS of $T[i..j]$ starting at $p_l-1$.
      Thus, $T[p_l-1..q-1]$ is repeating in $T[i..j+1]$.
      Therefore, $T[p_l-1..q]$ is a MUS of $T[i..j+1]$. \qed
  \end{itemize}
\end{proof}
\begin{proof}[of Theorem~\ref{thm:diff_MUS_add_right}]

  \noindent
  \begin{itemize}
    \item[(a)]
      By Lemma~\ref{lem:deleted_MUS}, $|\MUS(T[i..j]) \setminus \MUS(T[i..j+1])| \le 1$.
      By Observation~\ref{obs:suffix_MUS} and Lemma~\ref{lem:new_MUSs},
      $|\MUS(T[i..j+1]) \setminus \MUS(T[i..j])| \le 3$.
      Thus, $|\MUS(T[i..j+1]) \bigtriangleup \MUS(T[i..j])| =\
      |\MUS(T[i..j]) \setminus \MUS(T[i..j+1])| + |\MUS(T[i..j+1]) \setminus \MUS(T[i..j])| \le 4$.

      Next, we show that the upper bound is tight if $\sigma \ge 3$.
      For an integer $k \ge 2$, we consider two strings $u$ and $u'$ such that
      $u = \mathtt{a}^k\mathtt{b}\mathtt{c}\mathtt{c}$ of length $k + 3 \ge 5$ and
      $u' = u\mathtt{b} = \mathtt{a}^k\mathtt{b}\mathtt{c}\mathtt{c}\mathtt{b}$ of length $k + 4 \ge 6$.
      Then, $\MUS(u) = \{[1, k], [k+1, k+1], [k+2, k+3]\}$ and
      $\MUS(u') = \{[1, k], [k, k+1], [k+1, k+2], [k+2, k+3], [k+3, k+4]\}$.
      Therefore, $|\MUS(u') \bigtriangleup \MUS(u)| = 4$.
    \item[(b)]
      By Lemma~\ref{lem:deleted_MUS}, it is clear that $-1 \le |\MUS(T[i..j+1])| - |\MUS(T[i..j])|$.
      By Observation~\ref{obs:suffix_MUS}, the number of added MUS which is a suffix of $T[i..j+1]$ is at most one.
      Also, by Lemma~\ref{lem:new_MUSs}, the number of added MUS which is not a suffix of $T[i..j+1]$ is at most two,
      however, if such an added MUS exists, exactly one MUS~($= \SSuf_{i, j+1}$)
      must be deleted~(c.f. Lemma~\ref{lem:deleted_MUS} and Lemma~\ref{lem:MUS_contains_SS}).
      Therefore, $|\MUS(T[i..j+1])| - |\MUS(T[i..j])| \le 2$.

      Next, we show that each bound is tight if $\sigma \ge 3$.
      We consider strings $u$ and $u'$ that are described in the case~(a),
      and we then obtain $|\MUS(u')| - |\MUS(u)| = 2$.
      On the other hand, for any integer $\ell$ with~$\ell \ge 1$, we consider two strings $v$ and $v'$;
      $v = \mathtt{a}^\ell\mathtt{b}\mathtt{c}\mathtt{a}\mathtt{c}$ of length $\ell+4 \ge 5$ and
      $v = v\mathtt{a} = \mathtt{a}^\ell\mathtt{b}\mathtt{c}\mathtt{a}\mathtt{c}\mathtt{a}$ of length $\ell+5\ge 6$.
      If $\ell = 1$, then $\MUS(v)  = \{[2, 2], [3, 4], [4, 5]\}$,
      and                 $\MUS(v') = \{[2, 2],         [4, 5]\}$.
      If $\ell \ge 2$, then $\MUS(v)  = \{[1, \ell], [\ell+1, \ell+1], [\ell+2, \ell+3], [\ell+3, \ell+4]\}$,
      and                   $\MUS(v') = \{[1, \ell], [\ell+1, \ell+1],                   [\ell+3, \ell+4]\}$.
      Therefore, $|\MUS(v')| - |\MUS(v)| = -1$. \qed
  \end{itemize}
\end{proof}
\begin{figure}[ht]
    \centerline{\includegraphics[width=0.5\linewidth]{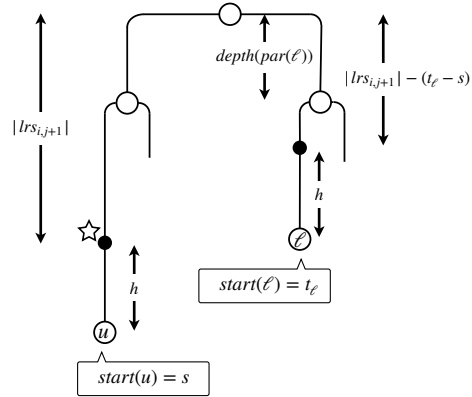}}
    \caption{
      For an example of Lemma~\ref{lem:leaf_has_implicit_node}.
      The situation of this figure is that
      each of $u$ and $\ell$ is a leaf with $t_{\ell} = \start(\ell) \ge s = \start(u)$
      and $|\LSuf_{i,j+1}| - (t_{\ell} - s) > \depth(\parent(\ell))$
      where $(u, h)$ represents the primary active point.
      Also, black nodes represent implicit suffix nodes.
    }\label{fig:implicit}
\end{figure}
\begin{proof}[of Lemma~\ref{lem:leaf_has_implicit_node}]
By Observation~\ref{obs:tp_occ}, for each leaf $\ell$, the suffix corresponding to
  the lowest implicit suffix node on $(\parent(\ell), \ell)$
  occurs exactly twice in $T[i..j+1]$ if such an implicit suffix node exists.

  Let $v = \LSuf_{i,j+1}$ and $\apoint_{i,j+1} = (u, h)$.
  If $u$ is not a leaf,
  there is no implicit suffix node on the edge~$(\parent(\ell), \ell)$ for any leaf $\ell$,
  since every suffix of $T[i..j+1]$ which is shorter than $|v|$
  occurs more than twice in $T[i..j+1]$.

  If $u$ is a leaf, $\numocc_{T[i..j+1]}(v) = 2$.
  Let $s = \start(u)$ and $t_{\ell} = \start(\ell)$ for each leaf $\ell$.
  In the case of $t_{\ell} < s$,
  we assume that there is an implicit suffix node on $(\parent(\ell), \ell)$ for the sake of contradiction.
  Let $w$ be a string corresponding to the lowest implicit suffix node on $(\parent(\ell), \ell)$.
  Then, $w$ is a suffix of $v$, and occurs exactly twice in $T[i..j+1]$.
  Furthermore, $w$ occurs exactly twice in $T[s..j+1]$.
  However, $w$ is a prefix of $T[t_{\ell}..j+1]$, hence $w$ occurs at least three times in $T[i..j+1]$, it is a contradiction.
  Thus, if $t_{\ell} < s$, there is no implicit suffix node on $(\parent(\ell), \ell)$.
Finally, we consider the case of $t_{\ell} \ge s$~(see Fig.~\ref{fig:implicit}).
  In this case, $T[t_{\ell}.. s+|v|-1]$ which is a prefix of $T[t_\ell..j+1]$ matches
  the suffix of $v$ which is $t_\ell-s$ characters shorter than $v$, i.e. $v[1+t_\ell-s..]$.
  Thus, there is an implicit suffix node on $(\parent(\ell), \ell)$ if and only if
  $|T[t_{\ell}.. s+|v|-1]| = |v|-(t_\ell-s) > \depth(\parent(\ell))$.
Also, if there is an implicit suffix node on $(\parent(\ell), \ell)$,
  the locus of the lowest one is $(\ell, h)$.
  \qed
\end{proof}
\begin{figure}[ht]
    \centerline{\includegraphics[width=0.8\linewidth]{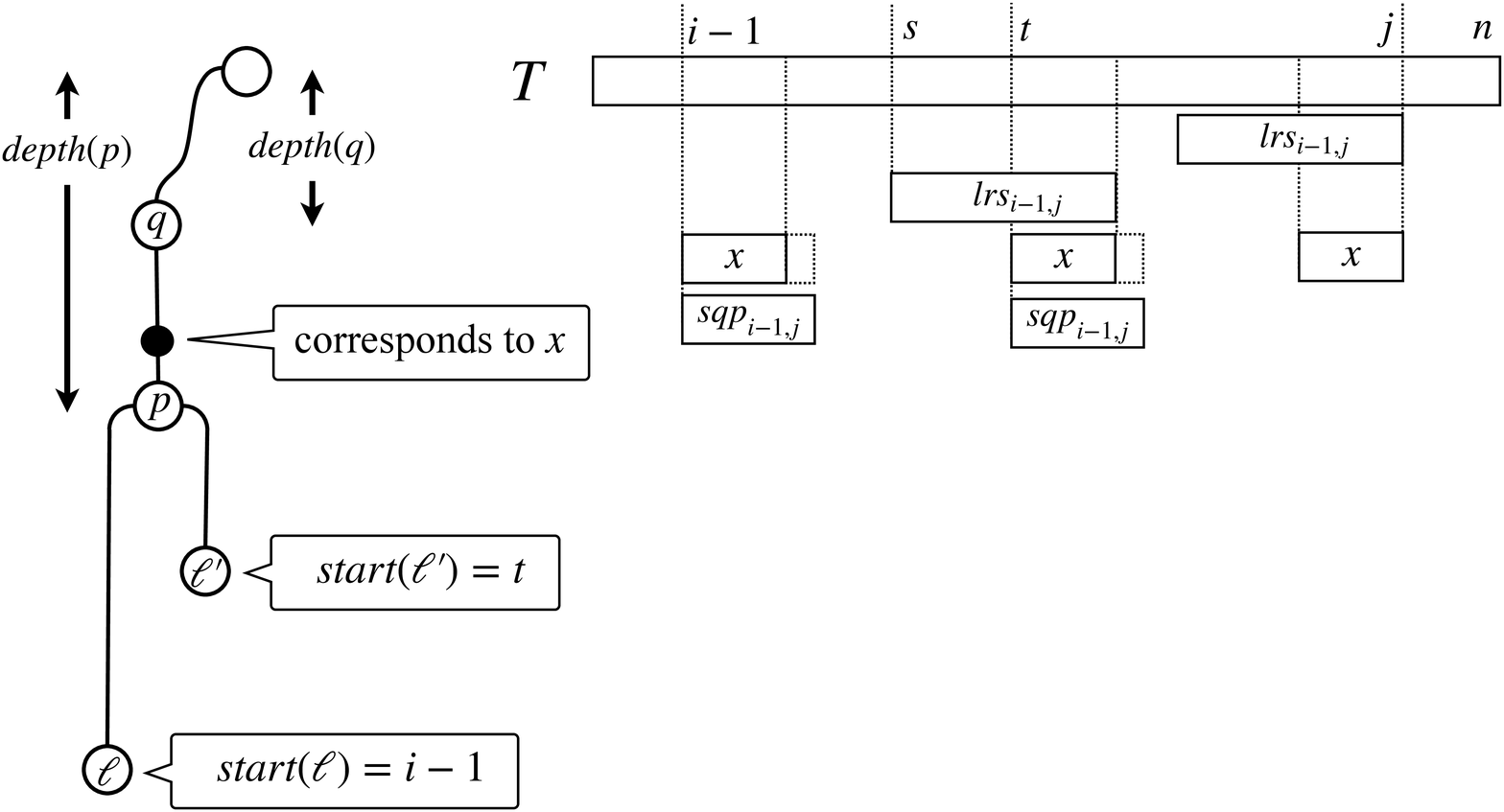}}
    \caption{
      Illustration for Lemma~\ref{lem:two_children}.
    }\label{fig:SP2}
\end{figure}
\begin{proof}[of Lemma~\ref{lem:two_children}]
Note that the suffix corresponding to the lowest implicit suffix node on $(q, p)$ occurs
  exactly three times in $T[i-1..j]$ from assumptions.
Let $\apoint_{i-1,j} = (u, h)$.
  If $h = 0$, the primary active point is an explicit node, and there is no implicit suffix node on every edge in $\STree_{T[i-1..j]}$.
  If $h \ne 0$ and $u = p$, the lowest implicit suffix node on $(q, p)$ is clearly the primary active point.
  Thus, in the following, we consider the situation with $u \ne p$ and $h \ne 0$.

  If $u$ is not a leaf and the number of leaves in $\subtree(u)$ is greater than two, then
  the number of leaves in $\subtree(\hed(v))$ is also greater than two for each implicit suffix node $v$.
  Thus, there is no implicit suffix node on $(q, p)$.
  If $u$ is not a leaf and the number of leaves in $\subtree(u)$ is exactly two, then
  $\LSuf_{i-1,j}$ occurs at least three times in $T[i-1..j]$ since $u \ne p$.
  Thus, if a suffix $s$ of $T[i-1..j]$ which is shorter than $\LSuf_{i-1,j}$ occurs as a prefix of $T[i-1..j]$,
  $\numocc_{T[i-1..j]}(s) \ge 4$, therefore, there is no implicit suffix node on $(q, p)$.

  If $u$ is a leaf,
  as in the proof in Lemma~\ref{lem:leaf_has_implicit_node}, it can be prove that
  there is an implicit suffix node on $(q, p)$ if and only if
  $t \ge s$ and $\depth(p) > |\LSuf_{i-1,j}|-(t-s) > \depth(q)$,
  where $s = \start(u)$, $t = \start(\ell')$ and $\ell'$ be the sibling of $\ell$~(see Fig.~\ref{fig:SP2}).
In addition, the length of the string $x$
  corresponding to the lowest implicit suffix node $\mathcal{X}$ on $(q, p)$
  is $|\LSuf_{i-1,j}|-(t-s)$, and thus, $\mathcal{X} = (p, \depth(p)-|x|) = (p, \depth(p) - |\LSuf_{i-1,j}| + t - s)$,
  if such an implicit suffix node exists.
  \qed
\end{proof}
\begin{proof}[of Lemma~\ref{lem:extention_total}]
  By Lemma~\ref{lem:extention_dec} and Lemma~\ref{lem:extention_inc}, we have
  $|\MAW(T[i..j+1]) \bigtriangleup \MAW(T[i..j])| = |\MAW(T[i..j+1]) \setminus \MAW(T[i..j])| + |\MAW(T[i..j]) \setminus \MAW(T[i..j+1])|\
  \le \sigma' + d + 1$.
In the following, we show that the upper bound is tight, i.e.
  there is a string $z$ of length $d$ and a character $\alpha$
  where $|\MAW(z) \bigtriangleup \MAW(z\alpha)| = \sigma' + d + 1$
  for any two integers $d$ and $\sigma'$ with $1 \le \sigma' \le d$ and $\sigma' + 1 \le \sigma$.
Let $\Sigma = \{a_1, a_2, \cdots, a_\sigma\}$ be an alphabet.
  Given two integers $d$ and $\sigma'$
  with $1 \le \sigma' \le d$ and $\sigma' + 1 \le \sigma$.
  We consider a string $z = a_1a_2 \cdots a_{\sigma'-1}a_{\sigma'}^{d-\sigma'+1}$ of length $d$
  and a character $\alpha = a_{\sigma'+1}$.
  Then, $\MAW(z) \setminus \MAW(z\alpha) = \{\alpha\}$.
  Also, $\MAW(z\alpha) \setminus \MAW(z) = \
    \{\alpha^2\} \cup \
    \{\alpha a_i\mid 1 \le i \le \sigma'\} \cup \
    \{ a_i\alpha\mid 1 \le i \le \sigma'-1\} \cup \
    \{a_{\sigma'-1}a_{\sigma'}^e\alpha\mid 1 \le e \le d-\sigma'\}$.
  Therefore, $|\MAW(z) \bigtriangleup \MAW(z\alpha)| = \sigma' + d + 1$.
  \qed
\end{proof}
\begin{proof}[of Lemma~\ref{lem:reduction_total}]
  Symmetric to the proof of Lemma~\ref{lem:extention_total}.
  \qed
\end{proof}
\begin{proof}[of Lemma~\ref{lem:sigma_slide_MAW}]
  If $\sigma = 2$, the lower bound
  $\mathcal{S}(T', d) \in \Omega(n-d) = \Omega(\sigma (n-d))$ is obtained
  by string $T' = (\mathtt{ab})^{n/2}$.

  In the sequel, we consider the case where $\sigma \ge 3$.
  Let $k$ be the integer with $(k-1)(\sigma-1) \le d < k(\sigma-1)$.
  Note that $k \geq 2$ since $\sigma \le d$.
  Let $\Sigma = \{a_1, a_2, \cdots, a_\sigma\}$ and $\alpha = a_\sigma$.
We consider a string $T' = U^e + U[..m]$ where
  $U = a_1 \alpha^{k-1} a_2 \alpha^{k-1} \dots a_{\sigma-1} \alpha^{k-1}$,
  $e = \lfloor \frac{n}{k(\sigma-1)} \rfloor$, and $m = n~\text{mod}~k(\sigma-1)$.
Let $c$ be a character that is not equal to $\alpha$.
  For any two distinct occurrences $i_1, i_2 \in \occ_{T'}(c)$ for $c$,
  $|i_1 - i_2| \ge k(\sigma-1) > d$.
  Thus, any character $c \ne \alpha$ is absent from at least one of two adjacent windows
  $T'[i..i+d-1]$ and $T'[i+1..i+d]$ for every $1 \le i \le n-d$.

  Now we consider a window $W = T[p-d..p-1]$ where
  $d + 1 \le p \le n$ and $T[p] = \beta \ne \alpha$.
  Let $\Pi = \{b_1, b_2, \cdots, b_\pi, \alpha\} \subset \Sigma \setminus \{\beta\}$
  be a set of all characters that occur in $W$.
  W.l.o.g., we assume that the current window is $W = \alpha^r b_1\alpha^{k-1}b_2\alpha^{k-1}\cdots b_\pi\alpha^{k-1}$
  and the next window is $W' = W[2..]\beta$
  where $r = d\text{~mod~}k$~(see Fig.~\ref{fig:MAW_total}).
  For any character $b \in \Pi \setminus \{b_1, b_\pi, \alpha\}$,
  $b \alpha^{\ell} \beta$ is in $\MAW(W')\bigtriangleup\MAW(W)$ for every $0 \le \ell \le k-1$.
  If $r > 0$,
  $b_1 \alpha^{\ell} \beta$ is also in $\MAW(W')\bigtriangleup\MAW(W)$ for every $0 \le \ell \le k-1$.
  Otherwise,
  $b_1$ is in $\MAW(W')\bigtriangleup\MAW(W)$ and
  $b_1\alpha^{\ell} b_2$ is in $\MAW(W')\bigtriangleup\MAW(W)$ for every $0 \le \ell \le k-2$
  since $b_1$ is absent from $W'$.
  Also, $\beta$ is in $\MAW(W')\bigtriangleup\MAW(W)$ and
  $b_\pi \alpha^{\ell} \beta$ is in $\MAW(W')\bigtriangleup\MAW(W)$
  for every $0 \le \ell \le k-2$.
  Thus, at least $(\pi-2) k + k + 1 + (k-1) = \pi k$ MAWs are in $\MAW(W')\bigtriangleup\MAW(W)$.
Additionally, the number $\pi$ of distinct characters which occur in $W$
  and are not equal to $\alpha$
  is at least $\lfloor (\sigma-1)/2 \rfloor$, since
  $k\lfloor(\sigma-1)/2\rfloor \le k(\sigma-1)/2 = (k - k/2)(\sigma-1) \le (k - 1)(\sigma-1) \le d$.
  Therefore, $|\MAW(W')\bigtriangleup\MAW(W)| \ge \pi k \ge \lfloor(\sigma-1)/2\rfloor k \in \Omega(\sigma k) = \Omega(d)$.
The number of pairs of two adjacent windows $W$ and $W'$
  where $|\MAW(W')\bigtriangleup\MAW(W)| \in \Omega(d)$
  is $\Theta((n-d)/k)$.
  Therefore, we obtain
  $\mathcal{S}(T', d) \in \Omega(d(n-d)/k) = \Omega(\sigma(n-d)) = \Omega(\sigma n)$
  since $n-d \in \Omega(n)$.
\begin{figure}[ht]
      \centerline{\includegraphics[width=0.8\linewidth]{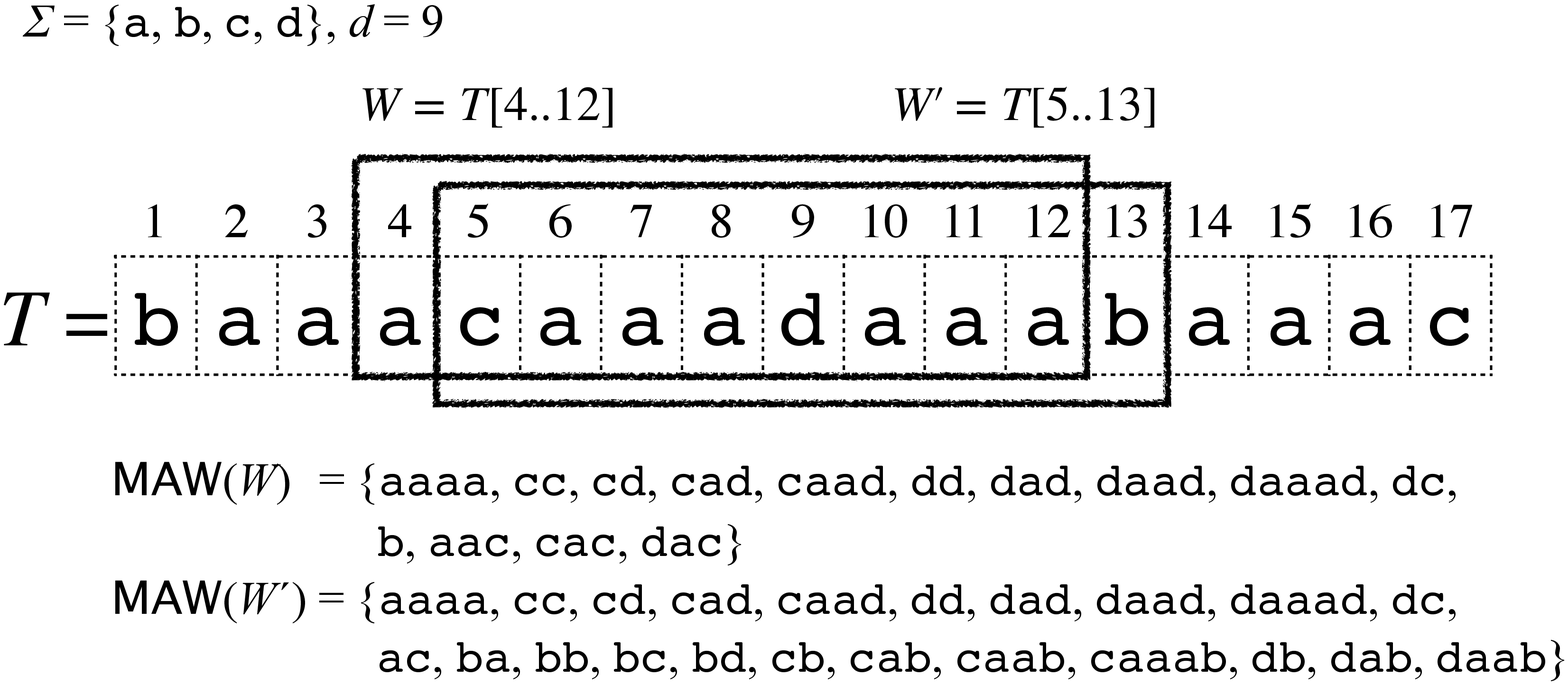}}
      \caption{
        Illustration of examples of MAWs in adjacent two windows.
        In this example, $\sigma = 4, d = 9 $, and $k = 4$.
        The size of the symmetric difference of $\MAW(W)$ and $\MAW(W')$ is
        $|\MAW(W) \bigtriangleup \MAW(W')| =$
        $|\{\mathtt{b},$ $\mathtt{aac},$ $\mathtt{cac},$ $\mathtt{dac},$ $\mathtt{ac},$
        $\mathtt{ba},$ $\mathtt{bb},$ $\mathtt{bc},$ $\mathtt{bd},$
        $\mathtt{cb},$ $\mathtt{cab},$ $\mathtt{caab},$ $\mathtt{caaab},$
        $\mathtt{db},$ $\mathtt{dab},$ $\mathtt{daab}\}| = 16$.
      }\label{fig:MAW_total}
  \end{figure}
\qed
\end{proof}
\begin{proof}[of Lemma~\ref{lem:d_slide_MAW}]
  By Corollary~\ref{col:increase_MAW_at_one_slide},
  it is clear that $\mathcal{S}(T, d) \in O(d(n-d))$.
  Next, we show that there is a string $T'$ of length $n > d$ such that
  $\mathcal{S}(T', d) \in \Omega(d(n-d))$
  for any integer $d$ with $1 \le d \le \sigma-1$.
Let $\Sigma = \{a_1, a_2, \cdots, a_\sigma\}$.
  We consider a string $T' = (a_1 a_2 \cdots a_{d+1})^e a_1 a_2 \cdots a_k$
  where $e = \lfloor n/(d+1) \rfloor$ and $k = n~\text{mod}~(d+1)$.
  For each window $W = T'[i..i+d-1]$ in $T'$,
  $W$ consists of distinct $d$ characters,
  and the character $T'[i+d]$ that is the right neighbor of $W$
  is different from any of characters occur in $W$.
  W.l.o.g., we assume that the current window is
  $W = a_1 a_2 \cdots a_d$
  and the next window is $W' = W[2..]a_{d+1}$.
  Then, $|\MAW(W') \bigtriangleup \MAW(W)| = \
  |\{a_{d+1}\} \cup \{a_{d+1}a_i\mid 2 \le i \le d\}| = d$.
  Therefore, $\mathcal{S}(T', d) = d(n-d)$.
  \qed
\end{proof}
\clearpage
\section{Figures} \label{sec:additional_fig}
In this appendix, we provide some supplemental figures.

\renewcommand{\floatpagefraction}{1.0}

\begin{figure}[pht]
  \centerline{\includegraphics[width=0.6\linewidth]{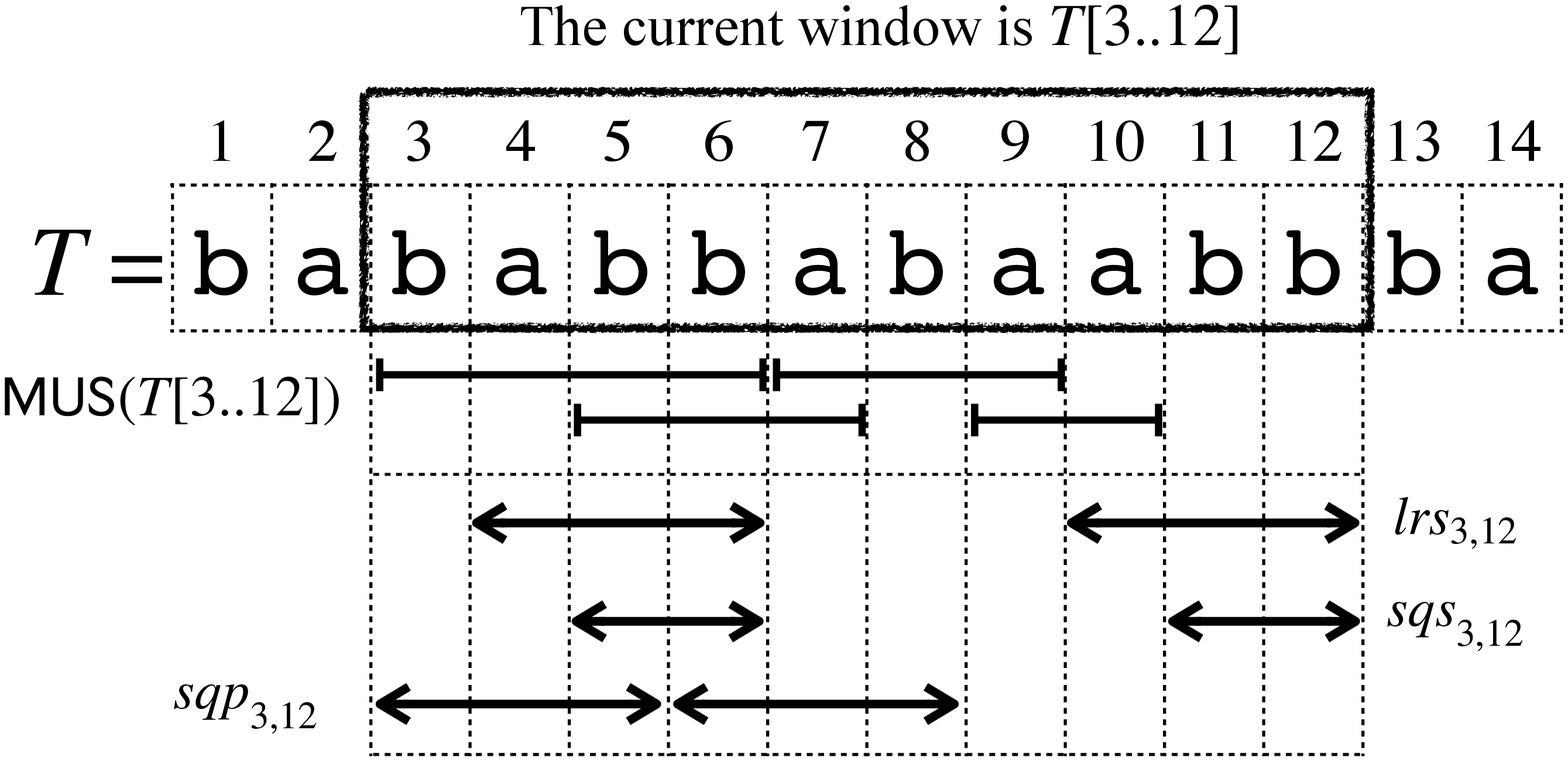}}
  \caption{
    String $T = \mathtt{bababbabaabbba}$ of length 14 and
    its substrings $\LSuf_{3,12}$, $\SSuf_{3,12}$, and $\SPref_{3,12}$
    for the current window $T[3..12]$.
  }\label{fig:LS_SS_SP}
\end{figure}
\begin{figure}[pht]
    \centerline{\includegraphics[width=0.6\linewidth]{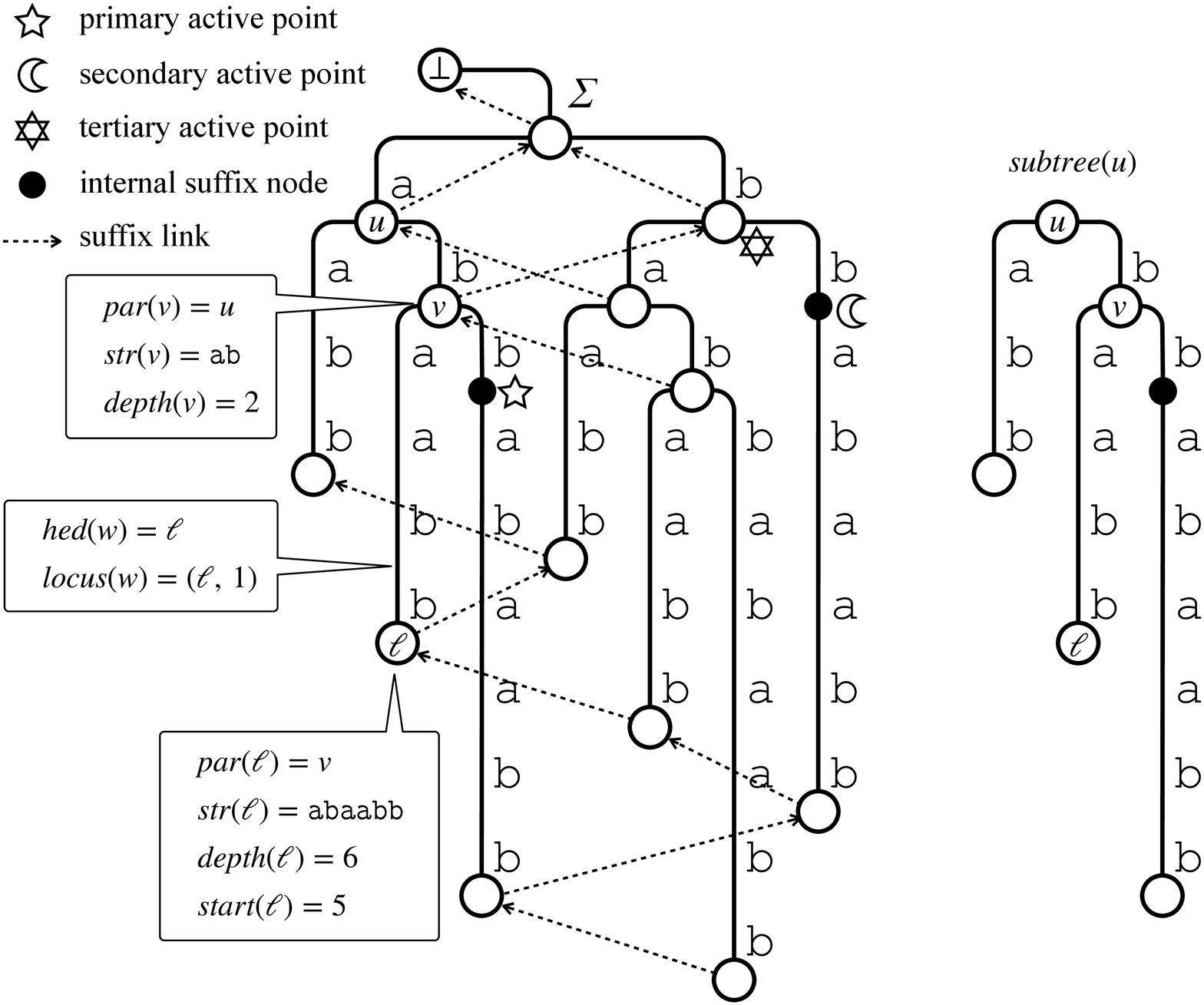}}
    \caption{
      The suffix tree of string $T = \mathtt{babbabaabb}$, where
      the suffix links are depicted by broken arrows,
      the implicit suffix nodes are depicted by black circles,
      as well as the three kinds of active points are marked.
      For example of other notions on the suffix tree,
      substring $w = \mathtt{abaab}$ of $T$ is considered here.
    }\label{fig:STree}
\end{figure}
\begin{figure}[ht]
    \centerline{\includegraphics[width=0.8\linewidth]{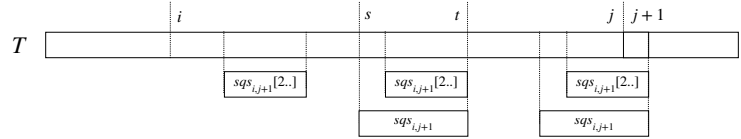}}
    \caption{
      Illustration for the case where
      $\numocc_{T[i..j+1]}(\SSuf_{i,j+1}) = 2$.
      In this case, $T[s..t] = \SSuf_{i, j+1}$ is unique in $T[i..j]$ and
      $T[s+1..t]$ is repeating in $T[i..j]$.
    }\label{fig:deleted_MUS}
\end{figure}
\begin{figure}[ht]
    \centerline{\includegraphics[width=0.8\linewidth]{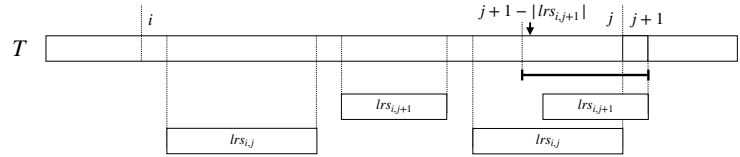}}
    \caption{
      Illustration for the case
      where $|\LSuf_{i, j+1}| \le |\LSuf_{i, j}|$.
      In this case,
      $T[j+1-|\LSuf_{i, j+1}|.. j+1]$ is a MUS of $T[i..j]$.
    }\label{fig:suffix_MUS}
\end{figure}
\begin{figure}[ht]
    \centerline{\includegraphics[width=0.6\linewidth]{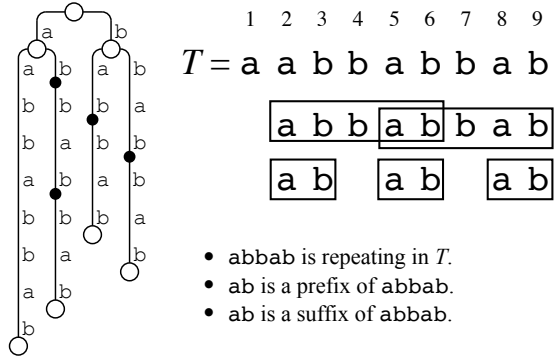}}
    \caption{
      The suffix tree of string $T = \mathtt{aabbabbab}$
      as an example of the Case~\ref{case_three}.1 in Observation~\ref{obs:tp_occ}.
      Black circles represent implicit suffix nodes.
      For two suffixes $s = \mathtt{ab}$ and $s' = \mathtt{abbab}$ of $T$,
      $\hed(s') = \hed(s)$ and $s$ occurs three times in $T$.
    }\label{fig:STree_occ}
\end{figure}
\begin{figure}[ht]
    \centerline{\includegraphics[width=0.8\linewidth]{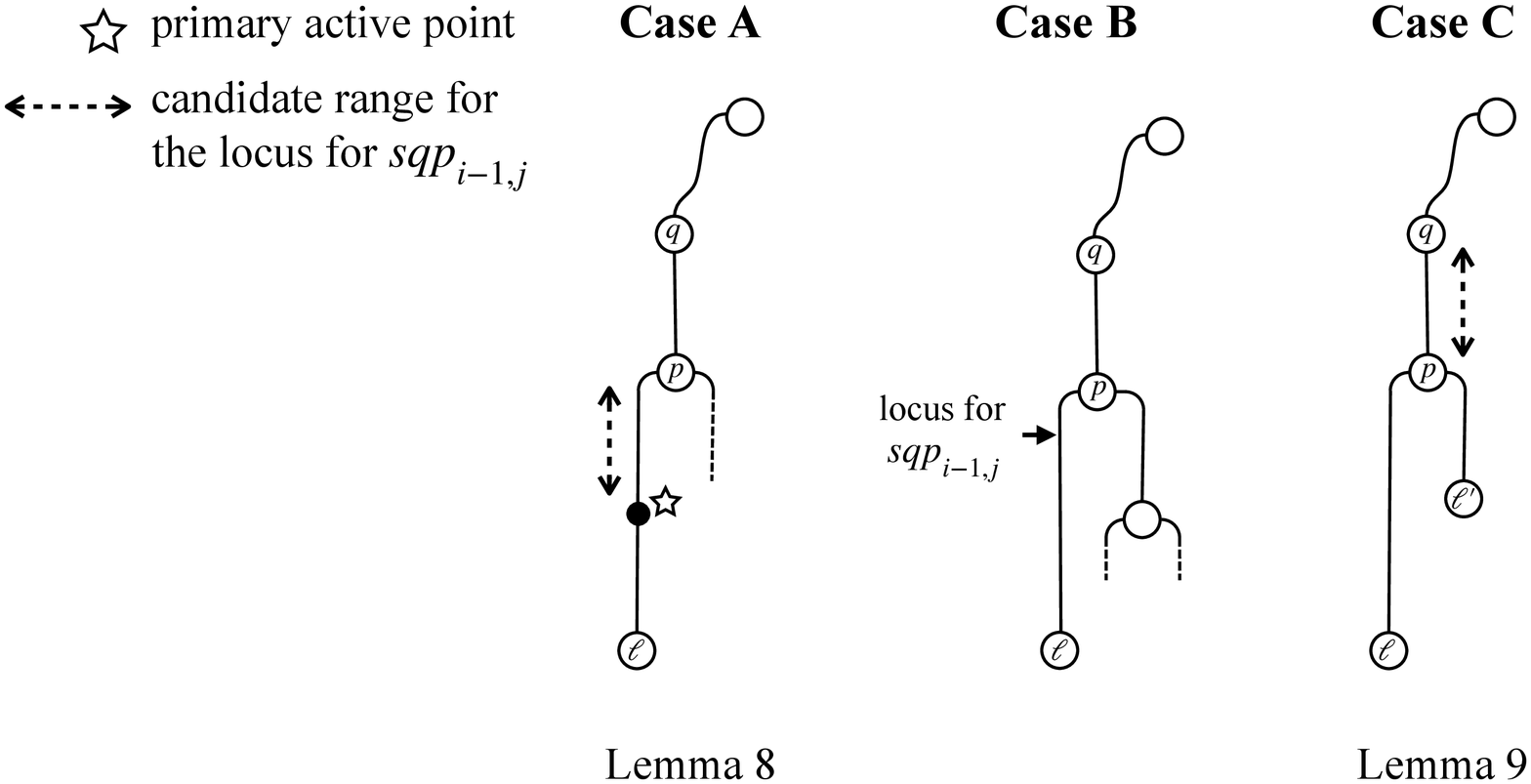}}
    \caption{
      Illustration for the three cases
      that are described in Section~\ref{subsec:sqs}.
    }\label{fig:SP_cases}
\end{figure}
\begin{figure}[ht]
    \centerline{\includegraphics[width=0.8\linewidth]{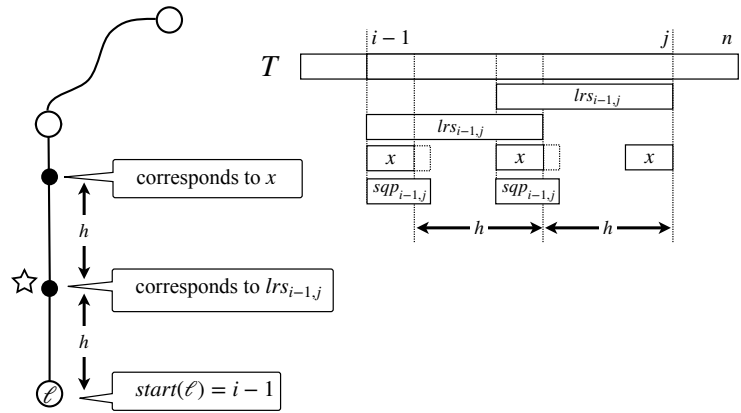}}
    \caption{
      Illustration for the proposition (c) in Lemma~\ref{lem:suffix_chain_come_back}.
    }\label{fig:SP}
\end{figure}
\begin{figure}[ht]
    \centerline{\includegraphics[width=0.8\linewidth]{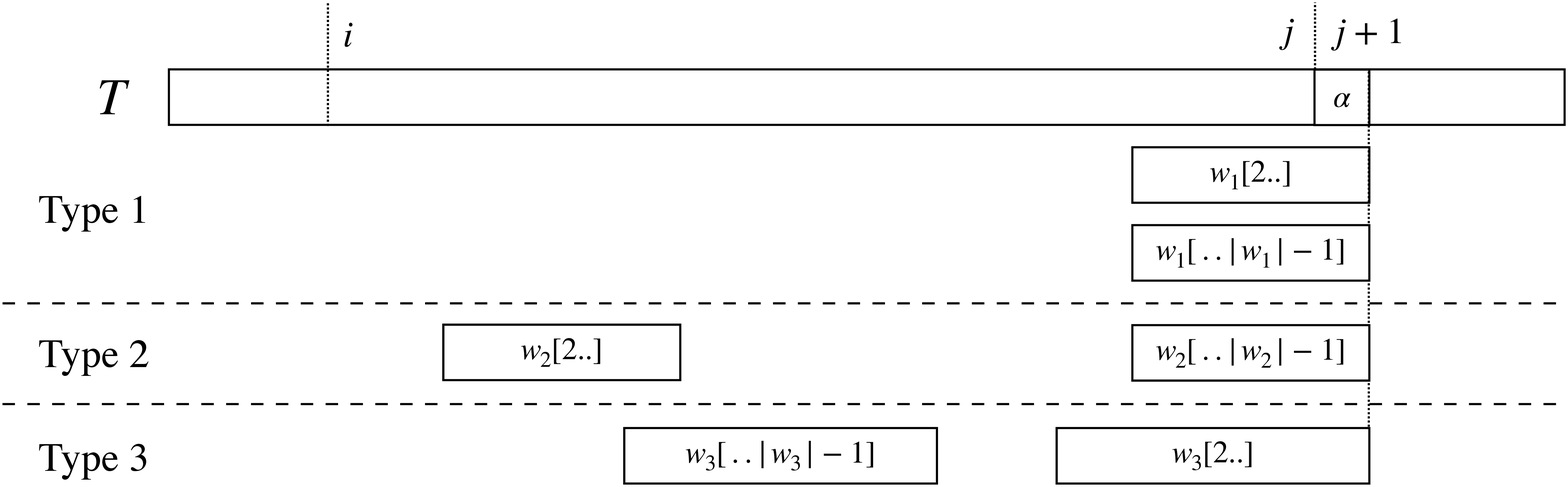}}
    \caption{
      Illustration for the three types of MAWs, where
      $w_1 \in \typeone$, $w_2 \in \typetwo$, and $w_3 \in \typethree$.
    }\label{fig:MAW_3types}
\end{figure}
\begin{figure}[ht]
    \centerline{\includegraphics[width=0.8\linewidth]{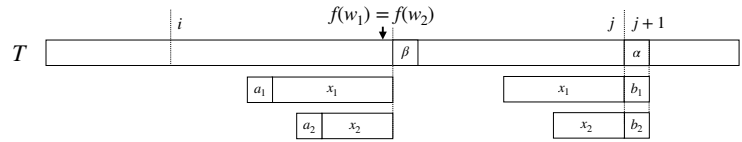}}
    \caption{
      Illustration for the contradiction in the proof of Lemma~\ref{lem:extention_inc}.
      Consider two strings $w_1 = a_1x_1b_1$ and $w_2 = a_2x_2b_2$ that are MAWs of $T$ of Type 3
      where $a_1, a_2, b_1, b_2 \in \Sigma$ and $x_1, x_2 \in \Sigma^\ast$.
      If $|w_1| > |w_2|$ and $f(w_1) = f(w_2)$, then $x_2$ is a proper suffix of $x_1$,
      and it contradicts that $a_2x_2b_2$ is absent from $T$.
    }\label{fig:MAW_Type3}
\end{figure}
 \end{document}